\documentclass[acmtog, screen, nonacm]{acmart}

\newtheorem*{claim*}{Claim}
\newtheorem*{corollary}{Corollary}
\usepackage{cases}    
\usepackage{cancel}
\usepackage{enumitem}

\newlist{alphalist}{enumerate}{1}
\setlist[alphalist,1]{label=\alph*.}

\AtBeginDocument{%
  \providecommand\BibTeX{{%
    \normalfont B\kern-0.5em{\scshape i\kern-0.25em b}\kern-0.8em\TeX}}}

\setcopyright{acmcopyright}
\copyrightyear{2024}
\acmYear{2024}

\acmJournal{TOG}
\acmVolume{37}
\acmNumber{4}
\acmArticle{111}
\acmMonth{8}

\acmSubmissionID{123-A56-BU3}

\newcommand{\R}{\mathbb{R}}
\usepackage{amsmath}

\usepackage{hyperref} 
\usepackage{algorithm}
\usepackage{algpseudocode}
\usepackage{eso-pic}
\usepackage{lipsum}  

\newtheorem{claim}{Claim}

\begin{document}

\AddToShipoutPicture*{ 
  \AtPageUpperLeft{ 
    \setlength\unitlength{1in} 
    \put(0.5,-0.5){ 
      \makebox(0,0)[left]{\textbf{Accepted by ACM Transactions on Graphics for 2024 publication. This document is the preprint version.}}
    }
  }
}

\title{Spectral Total-Variation Processing of Shapes - Theory and Applications}

\author{Jonathan Brokman}
\email{jonathanbrok@gmail.com}
\affiliation{%
  \institution{Technion}
  \streetaddress{Technion City}
  \city{Haifa}
  \country{Israel Institute of Technology}
  \postcode{32000}
}

\author{Martin Burger}
\email{martin.burger@desy.de}
\affiliation{%
  \institution{Helmholtz Imaging, Deutsches Elektronen-Synchrotron DESY}
  \streetaddress{Notkestr. 85}
  \city{Hamburg}
  \country{Germany}
  \postcode{22607}
}
\affiliation{%
  \institution{Fachbereich Mathematik, Universit\"at Hamburg}
  \streetaddress{Bundesstr. 55}
  \city{Hamburg}
  \country{Germany}
  \postcode{20146}
}
 
\author{Guy Gilboa}
\email{guy.gilboa@ee.technion.ac.il}
\affiliation{%
  \institution{Technion}
  \streetaddress{Technion City}
  \city{Haifa}
  \country{Israel Institute of Technology}
  \postcode{32000}
}

\renewcommand{\shortauthors}{Brokman et al.}

\begin{abstract}
We present an analysis of total-variation (TV) on non-Euclidean parameterized surfaces, a natural representation of the shapes used in 3D graphics. Our work explains recent experimental findings in shape spectral TV \cite{fumero2020nonlinear} and adaptive anisotropic spectral TV \cite{biton2022adaptive}. A new way to generalize set convexity from the plane to surfaces is derived by characterizing the TV eigenfunctions on surfaces. Relationships between TV, area, eigenvalue, eigenfunctions and their discontinuities are discovered. Further, we expand the shape spectral TV toolkit to include versatile zero-homogeneous flows demonstrated through smoothing and exaggerating filters. Last but not least, we propose the first TV-based method for shape deformation, characterized by deformations along geometrical bottlenecks. We show these bottlenecks to be aligned with eigenfunction discontinuities. This research advances the field of spectral TV on surfaces and its application in 3D graphics, offering new perspectives for shape filtering and deformation.
\end{abstract}




\maketitle

\begin{figure}[htbp]
\includegraphics[width=0.5\textwidth]{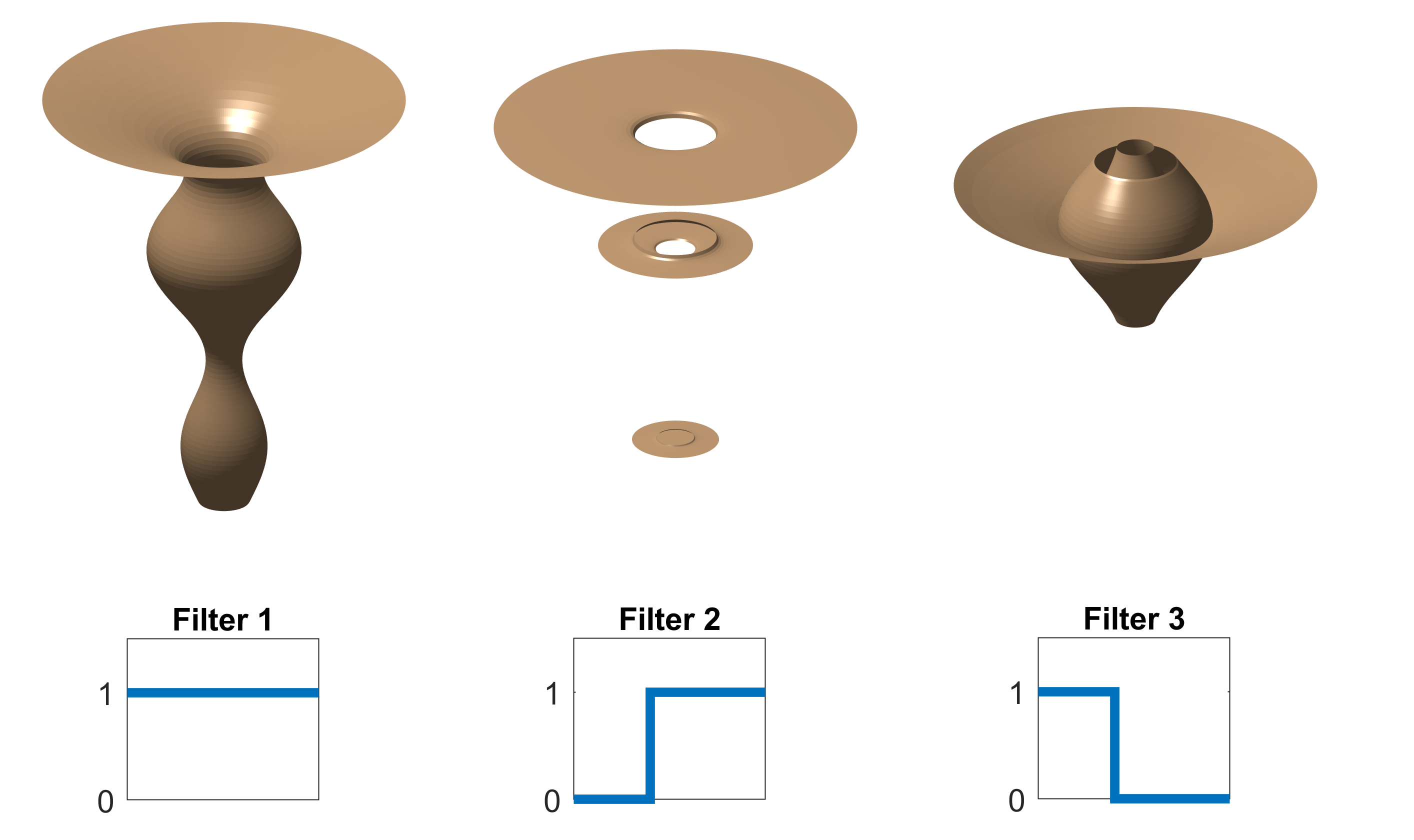} 
\caption{Spectral TV filtering of a shape, demonstrating the application of different filters to isolate TV eigenfunctions and reveal geometric details. Left: The original shape with an all-pass filter, displaying the complete spectrum. Middle: A low-pass filter uncovers the shape’s fundamental structure. From a theoretical point of view - it isolates eigenfunctions as lower-frequency components. Right: A high-pass filter brings forth the geometrical details on-top the foundational structure. The geometry is partitioned into three distinct subsets based on eigenfunction discontinuities (details in Fig. \ref{sinc}). This delineates the original shape by its bottlenecks, which are also observed in the boundaries of each of the three nested, concentric surfaces of the right panel. These bottleneck boundaries are a direct manifestation of the eigenfunction properties, which our manuscript investigates through novel theoretical and empirical analysis.}
\label{fig:sinc_hq}
\end{figure}

\section{Introduction}

Spectral geometry processing is a widely used technique in computer graphics. It involves breaking down shapes and their functions into spectral components, harnessing the multiscale nature and the well studied properties of eigenfunctions
\cite{taubin1995signal, sorkine2004laplacian, wetzler2013laplace, aflalo2013scale, bracha2020shape}. 
Spectral processing is useful for various computer graphics applications, such as mesh filtering, surface decoding, segmentation (spectral partitioning), shape deformation and general analysis. For a comprehensive exploration of this approach, including a survey and overview, please refer to  \cite{cammarasana2021localised}, and the references therein.

Generally, the spectral approach is based on linear algebra and harmonic analysis. It can be viewed as a generalization of Fourier basis to Riemannian manifolds and to graphs.
In recent years, there has been a growing branch  studying nonlinear spectral decompositions, see e.g. \cite{gilboa2018book}. In this setting, a signal is decomposed into spectral components related to nonlinear eigenfunctions of the form $\lambda u \in \partial J(u)$, where $\partial J$ is the subdifferential of a convex functional $J$. 
It can be readily seen that when $J$ is the Dirichlet energy, $J(u)=\|\nabla u\|_2^2 $, we have $\partial J(u)=\nabla J(u)= -\Delta u$, where $\Delta$ denotes the Laplacian. Thus, the eigen-decomposition reduces back to a linear Fourier-type analysis. In this sense, it is a generalization of linear spectral methods. 
Most of the research was dedicated to absolutely one-homogeneous functionals \cite{gilboa2014total:30,burger2016spectral,bungert2019nonlinear} such as total variation. These new signal analysis methods enabled crisp, edge-preserving decompositions and representations. However, the theory and applications were restricted mostly to images and signals in Euclidean domains. In this study we develop a theoretical setting for nonlinear spectral processing on surfaces, demonstrating the potential benefits of this approach to computer graphics. 

\begin{figure}[htbp]
\includegraphics[width=0.5\textwidth]{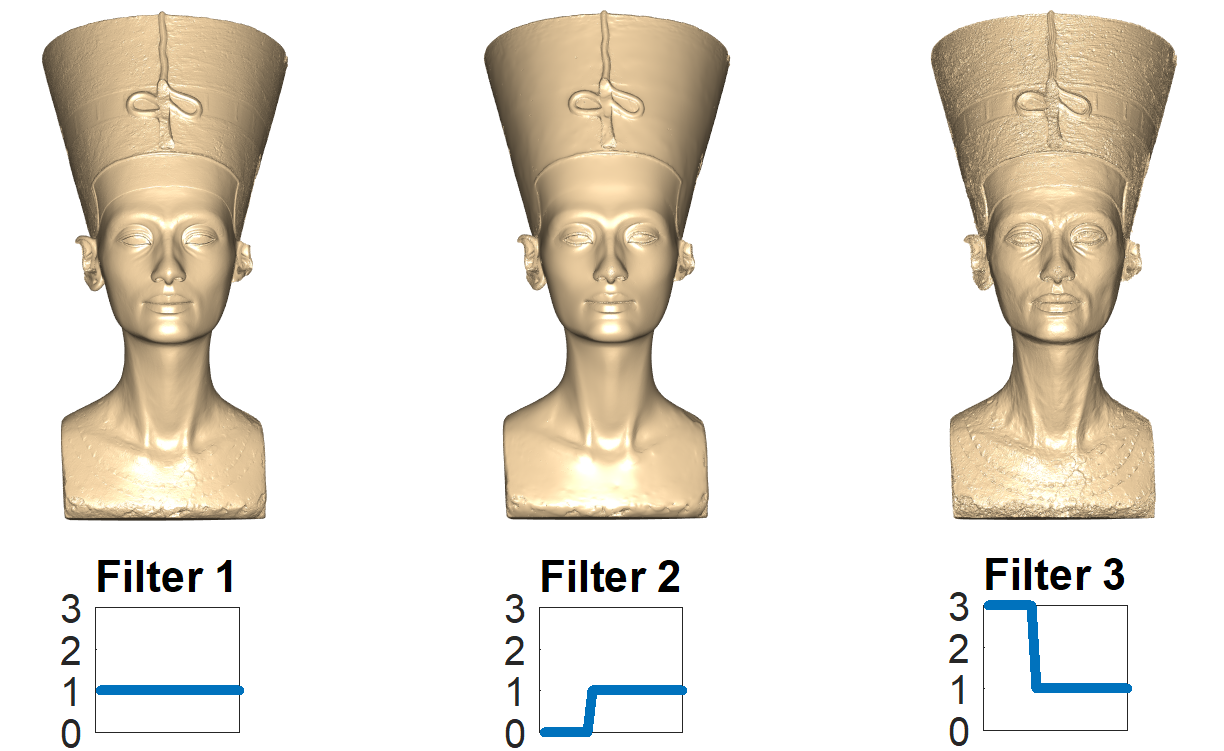} 
\caption{"Bust of Queen Nefertiti"\protect\footnotemark: Nonlinear spectral filtering applied using one of our proposed methods, presented in Sec. \ref{resmoothing method}}
\label{fig: nefertities}
\end{figure}

Total Variation (TV) is a popular regularization functional, well-known for its edge-preserving properties \cite{chambolle2004algorithm}. For a smooth function $u:\Omega \rightarrow \R$,  the TV functional is
\begin{equation}
TV(u)=\int_\Omega |\nabla u(x)|dx.      
\end{equation}
Total variation has found extensive applications in various domains and tasks over the past three decades.
Most notably, this approach was extensively used in image processing, for denoising, deconvolution, texture modeling, super-resolution, segmentation and more 
\cite{rof92,chambolle2011first,burger2013guide,wei2019regional,aujol2006structure,leng2021total,pascal2021automated, zhang2022robust}.  

\footnotetext{"Bust of Queen Nefertiti". Ägyptisches Museum und Papyrussammlung. The original sculpture was created in 1345 BC by Thutmose. Scanned by Nora Al-Badri and Jan Nikolai Nelles. \href{https://nefertitihack.alloversky.com/}{Available} under a \href{https://creativecommons.org/licenses/by-sa/4.0/}{Creative Commons license}.}

The exploration of Total Variation extends beyond images and encompasses shapes, point clouds, and graphs \cite{duan2019noise, sawant2020review, dinesh2020super}. Consequently, it is natural to harness the $TV$ functional for its nonlinear spectral properties in fields like spectral geometry. Such an approach was initiated by \cite{fumero2020nonlinear}, which introduced Spectral Total Variation \cite{gilboa2013spectral, Gilboa_spectv_SIAM_2014} for geometry analysis. Nonetheless, transitioning from images to shapes presents new challenges.

In the domain of image processing, regularization predominantly revolves around manipulating a function $u$ defined within a Euclidean domain. However, when it comes to geometry processing, the usage of $u$ as a representation of the processed shape introduces a significant distinction, 
as the assumption of a Euclidean metric 
is typically not valid. For Spectral TV, this distinction sparks unique and unexplored considerations.




The primary focus of this paper centers around the theory of  TV eigenfunctions and Spectral TV within the realm of 3D geometry. Our main contributions are:
\begin{itemize}
    \item A new way to generalize the definition of convex sets from planar domains to surfaces. This generalization arises as a property of total variation (TV) eigenfunctions on non-Euclidean surfaces, addressing theoretical inquiries of recent years \cite{fumero2020nonlinear}, \cite{biton2022adaptive}.
    \item Further quantitative relationships between eigenvalue, area, and the total variation of the eigenfunctions are derived. 
    \item Numerical demonstrations of our theoretical findings on the eigenfunction properties are provided, Figs. \ref{fig: spheres flow}, \ref{Toroid}, \ref{sinc}.
    \item Spectral TV is extended to include novel zero-homogeneous flows, utilized for spectral filtering of shapes, where each flow induces distinct filtering capabilities  - while being multi-scale as expected from spectral methods. See Figs. \ref{lion vase bands}, \ref{victoria michael limbs}, \ref{fig:4_armadil_caricature}.
    \item We present the first TV based solution to the Shape Deformation task, as described in the tutorial by  \cite{sorkine2009interactive}. This approach typically results in the concentration of deformation around geometrical bottlenecks, see Figs. \ref{fig:snail}, \ref{fig:additional_stretch}.
    \item The discovered properties of eigenfunctions are observed in both our shape filtering and shape deformation results, enhancing the understanding of these methods, notably in Figs. \ref{sinc}, \ref{armadil0 filtering}, \ref{fig:deformation_process},  \ref{fig:stretch_flow}. 
\end{itemize}

Through these contributions, we not only derive new generalizations of key theoretical concepts from Euclidean planes to non-Euclidean surfaces but also showcase how these theoretical advancements can be applied to shape processing tasks, offering new insights as well as methodologies. For 3D model details, see footnote under references. Our results are reproducible via our \href{https://drive.google.com/drive/folders/1QbfY7sR_Oonw3_0MFFGiRfRHmOJIAtrG?usp=sharing}{official implementation}.

\section{Previous Work}
Shape filtering has a long and rich history. The  pivotal work of \cite{taubin1995signal} proposed to utilize the shape-induced Laplacian eigenfunctions as a basis for shape filtering, in an analogue manner to classical signal processing techniques. A transform is computed by projecting the shape onto the basis, where filtering is obtained by weighted reconstruction via this basis. Many variations of this method were utilized for different tasks (e.g. \cite{sorkine2004laplacian}). Over time, the Laplace-Beltrami became the standard Laplacian of choice, for spectral applications, and in general \cite{wetzler2013laplace} 
\footnote{Such an adaptation of \cite{taubin1995signal} can be found e.g. in \cite{vallet2008spectral}, which proposed a computationally efficient shape filtering,  and demonstrated some core filtering capabilities: Shape exaggeration, detail enhancement, shape smoothing and regularization.}. For non-rigid shape processing, spectral representations of a scale-invariant version of the Laplace-Beltrami was introduced in \cite{aflalo2013scale}, and later combined with deep learning for shape correspondence \cite{bracha2020shape}, providing a non-rigid spectral analogue of  the popular spatial-domain approach of \cite{litany2017deep}.

 While Total Variation (TV) is traditionally used for image processing (see \cite{A2TV_Properties_chambolle2010introduction,BurgerOsher2013_TV_Zoo} for theory and applications), it was used in computer graphics as well,  for surface denoising, reconstruction, segmentation, and super-resolution \cite{zhong2018mesh, liu2017dirac, zhang2015variational, dinesh20193d, zhang2020total, kerautret2020geometric, dinesh2020super}. 
 Only recently, in 2020, a TV-based shape filtering approach was proposed \cite{fumero2020nonlinear}, advocating the application of spectral TV to the normals of the shape.

 TV and shape smoothing flows have a substantial background, spanning a considerable duration. The authors of \cite{kimmel1998image} established a close relation between TV image smoothing and Laplace-Beltrami shape smoothing. These techniques where based on discrete diffusion processes. Among a plethora of Laplacian-based diffusion processes we note two fundamental flows: the Mean Curvature Flow (MCF), and the Willmore flow, which gave rise to a new problem of singularities \cite{huisken1990asymptotic, blatt2009singular, desbrun1999implicit}. More stable variations of them where proposed throughout the years, for instance  \cite{crane2013robust} and \cite{kazhdan2012can}.  In \cite{elmoataz2008nonlocal} a generalization of the Laplacian-based approaches was proposed. Using graph-oriented operators they introduced $p$-Laplace flows for mesh fairing, including the 1-Laplace flow, which is a TV-flow, the gradient flow minimizing the TV energy. 

Various methods have been proposed to minimize numerically TV and a $L^2$ square fidelity term ($J_{TV}(u)+\lambda\|f-u\|_2^2$, often referred to as the ROF model \cite{rof92}), 
see for instance 
\cite{chambolle2004algorithm, chambolle2011first}. 
For implementing a TV-flow, one should note that the term $-\textrm{div}\large( \frac{\nabla u}{|\nabla u|}\large)$ is valid for the gradient estimation of the energy only for non-vanishing gradients of $u$.
Two common approaches are adopted to overcome this.
One is to use a regularized TV model
$J_{TV-\varepsilon}:= \int_\Omega \sqrt{|\nabla u(x)|^2+\varepsilon^2}dx.$ In this case, one does not obtain singularities at zero gradients and a straightforward explicit method can be used. Drawbacks of this approach are that an additional $\varepsilon$ parameter is introduced and that the flow is smoother (somewhat less edge preserving). Moreover, the evolution time step is highly restrictive, proportional to $\varepsilon$. Alternatively, one can approximate the gradient descent process as a series of proximal TV minimizations and use non-smooth solvers to obtain fully-edge preserving solutions, e.g. by \cite{chambolle2004algorithm} or \cite{chambolle2011first}. This yields more faithful results.
These methods can be regarded as semi-implicit methods, which are unconditionally stable (with respect to the time step parameter, as opposed to explicit methods).

In our work we adapt two such processes for  geometry processing: The first is the shape flow introduced in \cite{kazhdan2012can}, which proposed a  gradient descent semi-implicit flow  (see \cite{parikh2014proximal} for semi-implicit approaches).
It is adapted by changing their proposed gradient direction to conform with our setting.
The second is the iterative re-weighted L1 minimization scheme, introduced in \cite{bronstein2016consistent}, which is adapted by using a vectorial version of it, and again, selecting our operator of choice to replace their gradient direction. The shape deformation problem is classically solved as a constrained energy-minimization problem, for instance \cite{botsch2007linear, sorkine2007rigid}, which we also explore in our work.

Another approach which relates to our method is the family of shape processing techniques which process the displacement fields over a smooth version of the surface. This is mainly used for shape smoothing,  exaggeration and also for detail transfer, for instance \cite{cignoni2005simple, sorkine2009interactive, digne2012similarity}. Recently, \cite{yifan2021geometry} leveraged two separate neural networks for the over-smoothed shape and for the displacement.

Our most significant contribution is to the theory of spectral TV. Let us recap significant landmarks of this domain for the Euclidean setting.
 Spectral TV was introduced in \cite{gilboa2013spectral,Gilboa_spectv_SIAM_2014}, facilitating nonlinear edge-preserving image filtering. Essentially, the idea is to decompose a signal into nonlinear spectral elements related to eigenfunctions of the total-variation subdifferential. These nonlinear eigenfunctions raise important theoretical aspect - as their  properties enable an understanding of the behaviour of both TV regularization, and Spectral Total Variation processing.
The method is based on evolving gradient descent with respect to the TV functional, also known as the TV-flow \cite{tvFlowAndrea2001}. The spectral elements decay linearly in this flow.  Different decay rates correspond to different scales, where in the case of a single eigenfunction the rate is exactly the eigenvalue. Theoretical underpinning was performed for the spatially discrete case in \cite{burger2016spectral}, which also extended the concepts of spectral TV to decompositions based on general convex absolutely one-homogeneous functionals. Decompositions based on minimizations with the Euclidean norm, as well as with inverse-scale-space flows \cite{iss} were also proposed. The space-continuous setting was later analyzed in \cite{bungert2019nonlinear}. For the one-dimensional TV setting, it was shown that the spectral elements are orthogonal to each other. Various applications were suggested for image enhancement, manipulation and fusion \cite{BenningFusion2017,hait2019spectral}. A common thread related to gradient flows of one-homogeneous functionals is that they are based on zero-homogeneous operators. Note that other homogeneities are also possible, see for instance \cite{bungert2020asymptotic,cohen2020introducing}.

While in non-Euclidean domains the theory of TV eigenfunctions and spectral TV is not as developed as in the Euclidean case, other aspects of TV on non-Euclidean domains are highly researched.  A TV framework on manifolds was defined  in \cite{ben2007well}. This research was in the context of nonlinear hyperbolic conservation laws on Riemannian manifolds. The authors prove bounds and stability of the minimizing flow. Another example is anisotropic TV, which is typically applied for image processing (see \cite{grasmair2010anisotropic}), and can be analyzed as a functional on a non-Euclidean domain - as shown in \cite{biton2022adaptive}, which conducted experiments with spectral TV as well. In \cite{fumero2020nonlinear} experiments with geodesically convex sets\footnote{Geodesically convex sets are subsets of a manifold in which any two points have the geodesics between them contained in the set. Sometimes, it is further assumed that there is a unique geodesic between two such points.} were conducted. In their experiments, geodesically convex sets exhibited a linear decay throughout the minimizing flow. Theoretical explanations of this phenomenon were left as opened research questions. In the following we lay theoretical foundations for these  experimental results, validated by our proposed spectral nonlinear and non-Euclidean framework.





\section{Preliminaries} \label{sec: parametric surf} 
In this work, the processed shape is assumed to be a 2D manifold $M \subset \R^3$.  We assume $M$ is smooth, with a smooth parameterization
\footnote{$M$ locally behaves like $\R^2$. See for instance smooth manifold as "coordinate system" defined in  Thm. 5-2 of \cite{spivak2018calculus}. See also local diffeomorphism for instance in \cite{lee2013smooth}.}
$S(\omega_1, \omega_2) = (x(\omega_1, \omega_2), y(\omega_1, \omega_2), z(\omega_1, \omega_2))$, $\omega_1,\omega_2\in \Omega$,  i.e.  
\begin{equation}\label{eq: patmeteric manifold}
S:\Omega\subset\R^2\rightarrow M.
\end{equation}
Namely, $S$ is differentiable and invertible. In the following  we outline important well-known properties of this setting, which we use in our work. A celebrated resource for these properties  is  \cite{do2016differential}.

Let $f:M\rightarrow \R$, and $u = f\circ S(\omega_1, \omega_2)$ i.e. $u:\Omega\rightarrow\R$. Let $T_qM$ be the plane tangent to $M$ at point $S(\omega_1, \omega_2)= q \in M$. It can be shown that $\frac{\partial S}{\partial \omega_1}, \frac{\partial S}{\partial \omega_2} \in \R^3$, denoted $S_{\omega_1}, S_{\omega_2}$,  span $T_qM$ at point $q$, assuming they are linearly independent.   A field $F$ on $M$ assigns each $q$ with a vector in $T_q$, i.e. $F:M\rightarrow\cup(T_q)_{q \in M}$. 
Let $[t_1, t_2] \subset \mathbb{R}$ a connected interval. We denote a differentiable parametric curve $(\omega_1(t), \omega_2(t)), t \in [t_1, t_2]$ by $\tilde{\gamma}(t) : [t_1, t_2] \rightarrow \Omega$, and its mapped curve by $\gamma = S(\tilde{\gamma}(t)) \subset M$, i.e. $\gamma$ is a differentiable parametric curve which maps $\gamma:[t_1, t_2]\rightarrow M$. We write $\frac{d\gamma}{dt}$ as $\gamma_t$, often interpreted as the velocity at time $t$. Using the chain rule, $\gamma_t$ can be calculated as $\gamma_t = S_{\omega_1}\frac{\partial \omega_1}{\partial t} + S_{\omega_2}\frac{\partial \omega_2}{\partial t}= J\tilde{\gamma}_t$ where $J(\omega_1, \omega_2)$ is the Jacobian matrix with columns $(S_{\omega_1}, S_{\omega_2})$. For a vector $\tilde{a}$ originating from a point $\omega_1, \omega_2$,  similarly to the velocity vectors $\tilde{\gamma}_t$, $\tilde{a}$ is mapped to $M$ as
\begin{equation} \label{eq: Jacobian mapping}
a(q) = J(\omega_1, \omega_2)\tilde{a}, 
\end{equation}
where $q = S(\omega_1, \omega_2)$. Note that $a(q) \in T_q$. The induced inner product on $\Omega$  is $\langle\tilde{a}, \tilde{b}\rangle_g = (J\tilde{a})^T(J\tilde{b})=\tilde{a}^T(J^TJ)\tilde{b}$ . Hence,  when $M$ is parameterized using $(\Omega, S)$, it is equipped with the metric $g=J^TJ$:
\begin{equation} \label{eq: metric}
g(\omega_1,\omega_2) = \begin{pmatrix}
S_{\omega_1}^TS_{\omega_1} & S_{\omega_1}^TS_{\omega_2}\\
S_{\omega_2}^TS_{\omega_1} & S_{\omega_2}^TS_{\omega_2}
\end{pmatrix}.
\end{equation}
$S_{\omega_1}, S_{\omega_2}$ are assumed to be linearly independent, thus $g$ is positive definite (and invertible).  The vector magnitude  can now be calculated using $\tilde{a}, g$ as $||a||_2 = ||\tilde{a}||_g:=||J\tilde{a}||_2=\sqrt{\langle\tilde{a},\tilde{a}\rangle_g}$ . As a consequence, the length of $\gamma$  is $L(\gamma) = \int_c^d||\gamma_t||_2dt  = \int_c^d ||\tilde{\gamma}_t||_g dt$. 
The normal to $\gamma(t)$, denoted $n(t)$, is defined as the normalized intersection of two planes: perpendicular to $\gamma_t(t)$, and tangent to $M$ on $q=\gamma(t)$.

Let $C = S(\tilde{C})\subset M$ with $\tilde{C} \subset \Omega$. The perimeter of $C$ is the length of its boundary $\partial C$. Let $\gamma^C$ be a parameterized curve  mapping $t \in [c, d)$ to the whole of $\partial C$, then
\begin{equation} \label{eq: boundary length}
per(C) = \int_{\partial C} dl = \int_c^d ||\tilde{\gamma}^C_t||_g dt=L(\gamma^C),
\end{equation}
where the first equality expresses integration on the boundary regardless of the choice of parameterization, and $\gamma^C = S(\tilde{\gamma}^C)$.
The area of $C$ is
\begin{equation} \label{eq: submanifold area}
|C|= \int_C\,dM=\int_{\tilde{C}} \,da,
\end{equation}
where $da = \sqrt{|g|}d\omega_1d\omega_2$ is often referred to as an \emph{area element}, and the first equality expresses integration on the manifold, regardless of the choice of the parameterization of $S$ with respect to $\Omega$.

Let two functions $f_1, f_2$ be defined on the manifold $M$, with their corresponding representations in $\Omega$, $u_1, u_2$, respectively, defined in a similar manner to $u, f$ above. Throughout this work we consider functions to lie in a Hilbert space, commonly denoted as $\Gamma_{L_2}(M)$ (see for instance \cite{guneysu2015functions}), in which the inner product between two such functions is
\begin{equation} \label{function prod}
\int_Mf_1f_2dM=\int_{\Omega}u_1u_2da.
\end{equation}






Followingly, a gradient operator $\nabla_g$ which satisfies
$\langle\nabla_g u(\omega_1, \omega_2), \tilde{w}
\rangle_g = \lim\limits_{h \to 0}\frac{f(q+hw) - f(q)}{h}$, 
$\forall w = J\tilde{w}\in T_q$, with $||w||_2=||\tilde{w}||_g=1$ is obtained. The divergence operator $\nabla_g \cdot$ is then given by the adjoint of $\nabla_g$ , i.e. it satisfies,
\begin{equation} \label{eq: cont conjugate}
\int_\Omega \nabla_g\cdot \tilde{F} u  \,da = \int_\Omega \langle \tilde{F}, \nabla_gu \rangle_g \,da.
\end{equation}
Finally, we can define the divergence theorem on manifolds: 
 Let $C \subset M$ be a compact set  with a smooth boundary $\partial C$ and a boundary normal $n^C$. Then, $C$ may be treated as a manifold in its own right. Let $F$ be a vector field on $C$ ($F$ is compactly supported), then
 \begin{equation} \label{eq: div th}
 \int_C\nabla \cdot F \, dM = \int_{\partial C} F^Tn^C \, dl.
 \end{equation}
 We can express \eqref{eq: div th} also in the parameterization domain\footnote{This is usually stated in more general setting, for instance see proposition 4.9 in \cite{gallot1990riemannian}. For usage example see proof of Thm. 1 in \cite{fumero2020nonlinear}. 
 }:
\begin{equation}\label{eq: par div th}
\int_{\tilde{C}}\nabla_g \cdot \tilde{F} \, da = \int_{t_1}^{t_2}\langle \tilde{F}, \tilde{n}^C\rangle_g ||\tilde{\gamma}^C_t||_g\, dt,
\end{equation}
where $\tilde{\gamma}^C_t$ is defined as in Eq. \eqref{eq: boundary length}, as a smooth parametric curve along $\partial \tilde{C}$, i.e.  $\gamma^C = S(\tilde{\gamma}^C)$  is along $\partial C$.

The $\mathcal{P}$-Laplace-Beltrami is defined as
\begin{equation} \label{eq: g p Laplacian}
\Delta_{g,\mathcal{P}} (u):=\nabla_g \cdot (|\nabla_g u|^{\mathcal{P}-2} \nabla_g u).
\end{equation}
For $\mathcal{P}=2$, $\Delta_{g,\mathcal{P}} (u)$ coincides with the Laplace-Beltrami operator,  yielding a diffusion process on surfaces by,
\begin{equation} \label{eq: lin diff}
\frac{\partial u}{\partial t}=\Delta_{g,2} (u),\,\,u(0)=f.
\end{equation}
Other special cases we will discuss are $\mathcal{P}=3$, and $\mathcal{P}=1$.



\section{Parametric setting formulation}
\label{sec: parametric settings}

 In this section we present in detail the fundamentals of total variation on parametric surfaces. These serve our theoretical investigations and numerical demonstrations, presented later.
 
 We analyze functions on surfaces in the setting of Sec. \ref{sec: parametric surf}. 
 Let $M$ be a smooth manifold given as a differentiable and invertible parametric surface $S(\Omega)$, $\Omega \subset \R^2$ with domain variables $\omega_1,\,\omega_2 \in \Omega$. The Jacobian $J$ maps vectors from $\Omega$ to $M$, inducing the metric $g$, see Eqs. \eqref{eq: patmeteric manifold}, \eqref{eq: metric}.
 We assume to have a function $u:\Omega\rightarrow\R$, and a field $z:\Omega\rightarrow\R^2$. We assume that $z$  is differentiable with compact support. Note that $u$ is not necessarily continuous. We examine the following non-Euclidean TV functional,
\begin{equation} \label{eq: NETV} 
NETV(u) := \sup_{z}\int_{\Omega} u \nabla_g \cdot z \,da	\,\, s.t. \,\,||z||_g \leq 1 \,\forall \omega_1, \omega_2 \in \Omega.
\end{equation} 
This is a special case of TV on Riemannian manifolds, reduced to our parametric setting of two-dimensional surfaces (compatible with mesh processing). For the general case, see  \cite{miranda2003functions,ben2007well}.

$NETV(u)$ has two notable special cases: The Euclidean metric - which is obtained for a planar $M$, and the non-Euclidean integral formulation $\int_\Omega ||\nabla_gu||_gda$ - which is obtained for a differentiable function $u$. In this paper our focus is on non-Euclidean metrics, and non-continuous functions (unless otherwise noted).


Usually, a Neumann boundary condition is assumed, achieving invariance to shift by a scalar,
\begin{equation}\label{eq: addconst}
NETV( u + \alpha) =  NETV(u),\,\,\forall \alpha \in \R.
\end{equation}
Note that if $M$ is closed, then the boundary is empty - hence the Neumann boundary condition holds trivially.
Another commonly used property is that $NETV$ is absolutely one-homogeneous, i.e.
\begin{equation} \label{eq: NETV 1hom}
NETV(\alpha u) = |\alpha| NETV(u),\,\,\forall \alpha \in \R.
\end{equation}
If a field $z$ admits the supremum of $NETV(u)$, Eq. \eqref{eq: NETV}, then the field $sign(\alpha) z$ admits $NETV(\alpha u)$.
For an in-depth derivation and generalization of basic properties such as the above, see \cite{miranda2003functions}, Section 3.

Being absolutely one-homogeneous, we know that the  subdifferential of $NETV$, as stated for instance in \cite{burger2016spectral},  is the following set\footnote{The deduced subdifferential needs a product space. We consider the definition of Eq. \eqref{function prod}.}:

\begin{equation} \label{eq: 1hom differential}
\begin{gathered}
\partial NETV(u) =\\
\{p : NETV(v)\geq \int_\Omega p\,v\,  da \,\forall v, \,  NETV(u) = \int_\Omega p\,u\, da \},
\end{gathered}
\end{equation}
where $v$ is assumed to have the same Neumann conditions as $u$. By Eqs. \eqref{eq: 1hom differential}, \eqref{eq: NETV} we have

\begin{numcases}{\nabla_g \cdot z \in \partial NETV(u) \Rightarrow}
      NETV(u) = \int_\Omega u\nabla_g \cdot z\, da.\label{eq: subdiff to NETV} \\[1\baselineskip]
   NETV(v) \geq \int_\Omega v\nabla_g \cdot z\, da.\label{eq: other subdiff NETV},
\end{numcases}
and the converse
\begin{equation} \label{eq: converse subdiff}
\int_\Omega u\nabla_g \cdot z\, da = NETV(u) \Rightarrow \nabla_g \cdot z \in \partial NETV(u).
\end{equation}


The $NETV$ minimizing flow, performed on a function $f:\Omega \rightarrow \mathbb{R}$ , is defined as
\begin{equation} \label{eq: scalar flow}
u_t = -p(t), \quad p(t) \in \partial NETV(u(t) , \quad u(0) = f, \,\, t\geq 0, 
\end{equation} 
where $u_t=\frac{\partial u} {\partial t}$.
Recently, \cite{bungert2020asymptotic} proved that eigenfunctions are exposed upon decay by such flows as asymptotic solutions (just before extinction). This enables us to reveal eigenfunctions numerically, by simulating Eq. \eqref{eq: scalar flow}, as done in Figs.  \ref{sinc_set}, \ref{fig: spheres flow}, \ref{Toroid}.



\begin{figure} [htbp]
\includegraphics[width=0.5\textwidth]{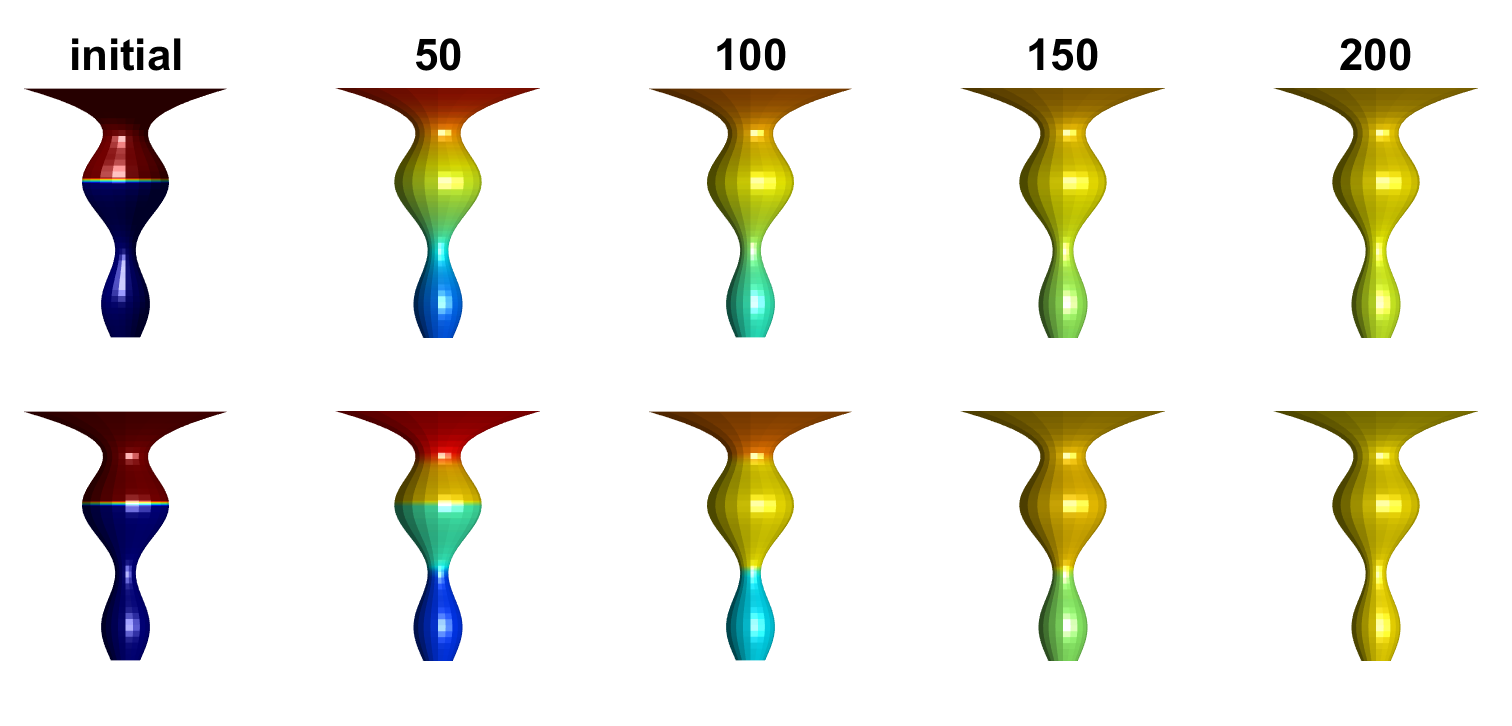}
\caption{$M$ is a surface-revolution of a part of a translated sinc curve, inducing a non-Euclidean metric. $f$ is initialized as an indicator function of a "sleeve set". Upper row: linear diffusion. Bottom row: $NETV$ minimizing flow. Unlike linear diffusion, the function remains piecewise constant throughout the flow -  dividing the surface $M$ to subsets. Iter 50: New boundaries, of small perimeters emerge; iter 100: Initial boundaries subside; iter 150: The sets merge so that only the minimal perimeter remains - which is numerically shown to be an eigenfunction (by Thm. 2.3 in \cite{bungert2020asymptotic}).  When an eigenfunction indicates a set (as is the case here), we call it an eigenset. In the Euclidean case - the eigensets' shape and behaviour have well studied properties \cite{andreu2001minimizing, bellettini2002total}. In contrast, eigensets of the non-Euclidean case are less understood . We introduce novel theoretical properties of the eigensets in non-Euclidean domains.}
\label{sinc_set}.
\end{figure}

\subsection{Indicator Functions}
\label{sec: indicator theory}
Indicator functions are non-continuous functions which have an important role in total variation analysis. In our non-Euclidean setting we analyze the indicator function of a subset $C \subset M$,
\begin{equation}
\chi^C(q) = \begin{cases}
      1 \, &  q \subset C \\
      0 & q \subset {M \backslash C},
    \end{cases}
\end{equation}
for which we construct 
\begin{equation}\label{eq: chi tildee}
\tilde{\chi}^C = \chi^C \circ S,
\end{equation}
i.e. $\tilde{\chi}^C$ is the indicator of $\tilde{C} \subset \Omega$ where %
\begin{equation} \label{eq: C tilde}
\tilde{C} = \{\,S^{-1}(q)\,|\,q\in C\,\}. 
\end{equation}
For convenience, let us define the $NETV$ of a set as the $NETV$ of its indicator function $NETV(C) := NETV(\tilde{\chi}^C)$, i.e.
\begin{equation}\label{eq: NETV of C}
NETV(C) := \sup_{z}\int_{\tilde{C}} \nabla_g \cdot z \,da	\,\, s.t. \,\,||z||_g \leq 1 \,\forall \omega_1, \omega_2 \in \Omega.
\end{equation}
In the following we assume the setting in which the divergence Thm. on manifolds \eqref{eq: par div th} holds, i.e.: $C$ is a connected set with a smooth boundary $\partial C$. The boundary has normals  $n^C$ (perpendicular to $\partial C$ and tangent to $M$) with a corresponding $\tilde{n}^C$ s.t. $n^C=J\tilde{n}^C$, where $J$ is the Jacobian from Sec. \ref{sec: parametric surf}. In addition, we assume a parameterized curve $\gamma^C:t \in [t_1, t_2)\rightarrow\partial C$ for which we construct $\tilde{\gamma}^C = S^{-1} \circ\gamma^C$, a parameterized curve along $\partial \tilde{C}$\footnote{To understand how $S$ maps boundaries from one domain to another refer to \cite{lee2013smooth}, Thm. 2.18.}.

Let a field $z$ that is normal to the boundary of $C$ on (almost all) boundary points, and of norm less than or equal to one everywhere on $M$, i.e.
 \begin{equation}
 z=\tilde{n}^C \,\,for \,a.e.\, \omega_1, \omega_2 \in \partial \tilde{C},\,\,\,
 ||z||_g \leq 1 \,\forall \omega_1, \omega_2 \in \Omega,
 \end{equation}
where $a.e.$ stands for "almost every". Then $z$ admits the supremum of Eq. \eqref{eq: NETV of C} for $NETV(C)$. Such a z exists if the boundary of $C$ is differentiable almost everywhere. This condition is satisfied by the assumption of a smooth boundary.

Last property but not least, we have that
\begin{equation}
NETV(C)=per(C).    
\end{equation}
For a proof and further details see Sec. \ref{app:indicator_proofs} in the Appendix.


\section{Theoretical Findings\label{sec: thry}}
Here we generalize the theory from the Euclidean setting to closed non-Euclidean manifolds\footnote{Note that closed manifolds induce a different boundary condition than the non-closed case. The treatment of both cases is similar - for brevity we show the closed case only.} of the parametric surface setting (Sec. \ref{sec: parametric settings}). Ultimately a new generalization of convexity is derived, as we prove properties of the eigenfunctions of the sub-gradient. The theory is demonstrated numerically as well.

\subsection{Derivations}

 \begin{definition}\label{def: eigenfunction}
 $u$ is an eigenfunction of $NETV$  if  $\exists \xi: \Omega\rightarrow \R^2$ s.t.
 \begin{alphalist}
 \item  $\lambda u \in \partial  NETV(u)$
\item $\nabla_g \cdot \xi = \lambda u$
 \end{alphalist}
 for some $\lambda \in \R$.
 \end{definition}
 
 \begin{definition}
 $C$ is an eigenset if $\exists \beta>0$ s.t.
 \begin{equation} \label{eq: psi def}
 \psi:=\tilde{\chi}^C - \beta \tilde{\chi}^{{M\backslash C}}
 \end{equation}
is an eigenfunction of $NETV$, where  $M\backslash C$ is the complement of $C$, and $\tilde{\chi}^{{M \backslash C}}$ is defined as in Eq. \eqref{eq: chi tildee}.
 \end{definition}
 The existence of eigensets should be verified for each $M$ individually. We provide an example of such verification for a torus in the appendix. We also find eigensets for various manifolds in our numerical experiments.

\begin{figure}[htbp]
\includegraphics[width=0.5\textwidth]{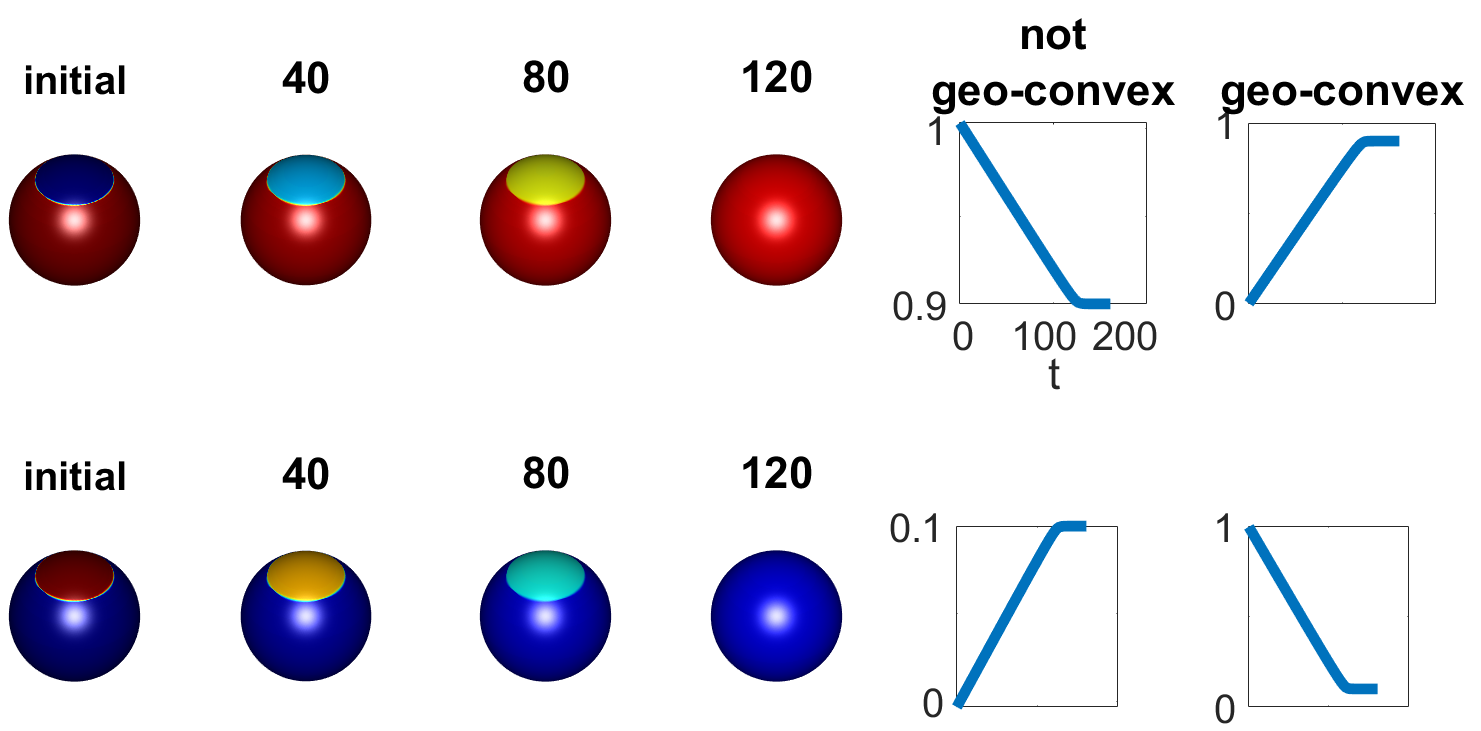} 
\caption{$NETV$ minimizing flow on a spherical manifold $M$. The initial function $f$ assumes two options: An indicator function of a geodesically convex spherical cap $C$ (lower row), or of $M \backslash C$, which is not geodesically-convex (upper row). From right to left: Values of $q_1 \in C$ throughout the flow, Values of $q_2 \in M \backslash C$ throughout the flow, and the flow portrayed as color on $M$. $C$ and $ M \backslash C$ remain intact throughout the flow, and their values change linearly with time $t$, until they decay completely. Such a behaviour implies that both $f$s are eigenfunctions, i.e. $C$ and $M \backslash C$ are eigensets. Noteably $M \backslash C$ is not geodesically convex. However, in the Euclidean case, eigensets must be convex sets. This raises the question: Is there an alternative notion of convexity on manifolds, other than geodesical convexity, which characterizes eigensets similarly to the Euclidean case? In the following we define and prove such a notion.}
\label{fig: spheres flow}
\end{figure}


\begin{claim}  
If $C$ is an eigenset with $\psi$ as in \eqref{eq: psi def} and $\lambda \neq 0$, then

\begin{equation} \label{eq: beta}
\beta=\frac{|C|}{|{M \backslash C}|},
\end{equation}
where $|C|, |{M \backslash C}|$ are areas of  $C, {M \backslash C}$.

\end{claim}

\begin{proof}
By manifold divergence theorem (see for instance proposition 4.9 in \cite{gallot1990riemannian}) we have for closed manifolds that $\int_\Omega \nabla_g \cdot z da=0$, since their boundary is an empty set\footnote{See for instance chapter 4.A.2 in\cite{gallot2004differential}}. $\psi$ is an eigenfunction, i.e. $\exists \xi$ as in Def. \ref{def: eigenfunction}, namely $\nabla_g \cdot \xi = \lambda \psi$. Plugging this we have $\int_\Omega \nabla_g \cdot \xi \,da=\int_\Omega \lambda \psi \,da=0$. Plugging definition \eqref{eq: psi def} we get
\begin{equation}\label{eq: lambda integral}
\lambda\int_\Omega\tilde{\chi}^C-\beta\tilde{\chi}^{{M \backslash C}}da=0.
\end{equation}
Thus, since $\lambda \neq 0$, by Eq. \eqref{eq: submanifold area}
\begin{equation}
|C|-\beta |{M \backslash C}|=0.
\end{equation}
\end{proof}
We note that in case $C=M$, we have by Eq. \eqref{eq: lambda integral} that $\lambda|M|=0$, i.e. $\lambda=0$. This settles well with the following claim.

\begin{claim}\label{cl: eigenset area tv}
If $C$ is an eigenset, then
\begin{equation} \label{eq: eigenvalue NETV area}
\lambda|C| = NETV(C).
\end{equation}
\end{claim}
\begin{proof}

Let $\psi$ as in \eqref{eq: psi def}. Plugging the eigenfunction property $\nabla_g \cdot \xi = \lambda \psi$ to $\int_\Omega\psi\nabla_g \cdot \xi \,da$ yields
\begin{equation} \label{eq: on one}
\int_\Omega\psi\nabla_g \cdot \xi \,da = \lambda \int_\Omega \psi^2 \,da = \lambda (|C| + \beta^2|{M \backslash C}|) = \lambda |C|  (1+\beta),
\end{equation}
where the last equality uses Eq. \eqref{eq: beta} as follows: $ |C| + \beta^2|{M \backslash C}|$ $  =  |C|(1+\beta^2\frac{|M \backslash C|}{|C|}) $ $= |C|  (1+\beta)$. On the other hand, we use Def. \ref{def: eigenfunction} again, namely $\nabla_g \cdot \xi \in \partial NETV(\psi)$ - which we plug to Eq. \eqref{eq: subdiff to NETV} and get
\begin{equation} \label{eq: on other}
\int_\Omega\psi \nabla_g \cdot \xi \,da = NETV(\psi) = NETV((1+\beta)\tilde{\chi}^C -  \beta) = (1+\beta)NETV(C),
\end{equation}
where the last equality is by Eqs. \eqref{eq: NETV 1hom}, \eqref{eq: addconst}, and the equality before uses $\psi =\tilde{\chi}^C - \beta \tilde{\chi}^{{M\backslash C}}=(1+\beta)\tilde{\chi}^C -  \beta$. Thus we have by Eqs. \eqref{eq: on one} and \eqref{eq: on other} the relation
\begin{equation} 
\int_\Omega\psi \nabla_g \cdot \xi \,da =\lambda |C|  (1+\beta)= (1+\beta)NETV(C),
\end{equation}
i.e.
\begin{equation}
 \lambda |C| =  NETV(C).
\end{equation}
\end{proof}
\begin{corollary}  

Plugging Eq. \eqref{eq: indictr NETV} to Eq. \eqref{eq: eigenvalue NETV area} we have

\begin{equation}\label{eq: eigenvalue perimeter area}
 \lambda |C| =  NETV(C)=per(C).
\end{equation}

\end{corollary}

\begin{definition}
A set $C\subset M$ with a nonempty $\partial C$ is called a \emph{minimal perimeter set} if every set $D$ that contains $C$ has a larger perimeter, i.e. 
\begin{equation} \label{eq: minimal perimeter}
per(C) \leq per(D) \,\, \forall D \supset C.
\end{equation}
 \end{definition}
 
 \begin{definition}
A set $C \subset M$ with a nonempty $\partial C$ is called a \emph{locally minimal perimeter set} if 
 \begin{equation} \label{eq: generalized minimal perimeter}
per(C) \leq \frac{|{M \backslash C}|}{|{M \backslash D}|} per(D) \,\, \forall D \supset C.
\end{equation}
 \end{definition}
 This definition is a weaker than the previous definition, since $\frac{|{M \backslash C}|}{|{M \backslash D}|} > 1$.
 The locality of this definition can be explained as follows: The larger $\frac{|{M \backslash C}|}{|{M \backslash D}|} $ is, the further away some points in $D$ must be from the local neighbourhood of $C$. 
 Noteable cases:
 \begin{itemize}
 \item  Any minimal perimeter set is also a locally minimal perimeter set.
\item Closed surfaces have $per(M)=0$ thus they do not contain minimal perimeter sets, but they do contain {\bfseries locally} minimal perimeter sets.
\item In the Euclidean case, both these definitions coincide with the set being convex.
 \end{itemize}

 \begin{theorem} \label{cl: locally minimal NETV}  Any eigenset $C$, a subset of a closed $M$, satisfies
\begin{equation} 
NETV(C) \leq \frac{|{M \backslash C}|}{|{M \backslash D}|} NETV(D), \,\, \forall D \supset C.
\end{equation}
\end{theorem}

\begin{proof}
By the eigenset Def. \ref{def: eigenfunction}, $\exists \psi$ as in \eqref{eq: psi def} with a field $\xi$ s.t. $\nabla_g \cdot \xi \in \partial NETV(\psi)$. Thus by Eq. \eqref{eq: other subdiff NETV} we have
\begin{equation}
\int_{\Omega} \tilde{\chi}^D\nabla_g \cdot \xi \,da \leq NETV(\tilde{\chi}^D),
\end{equation}
where $\tilde{\chi}^D$ is defined similarly to $\tilde{\chi}^C$. By definition of $\tilde{\chi}^D$ we have $\int_{\Omega} \tilde{\chi}^D (\cdot) \,da=\int_{\tilde{D}} (\cdot)\,da$. Using this and the notation of Eq. \eqref{eq: NETV of C}, we have
\begin{equation}\label{eq: on one ineq}
\int_{\tilde{D}}\nabla_g \cdot \xi \,da \leq NETV(D).
\end{equation}
On the other hand, we have the eigenset property  $\nabla_g \cdot \xi = \lambda \psi$. Plugging this to $\int_{\tilde{D}}\nabla_g \cdot \xi \,da$, we have
\begin{equation}\label{eq: on other ineq}
\int_{\tilde{D}}\nabla_g \cdot \xi \,da = \lambda \int_{\tilde{D}}\psi \,da = \lambda(|C| -\beta(|D|-|C|)).
\end{equation}
Note that we integrate over $D$, while $\psi$ is defined for $C$, $M$. The last equality uses $C \subset D$ which translates to $\tilde{C} \subset \tilde{D}$ by invertability of $S$. Thus by Eqs. \eqref{eq: on one ineq}, \eqref{eq: on other ineq}
\begin{equation}
\lambda(|C| -\beta(|D|-|C|)) \leq NETV(D).
\end{equation}
Plugging Eqs. \eqref{eq: eigenvalue NETV area}, \eqref{eq: beta} we get
\begin{equation}
\frac{NETV(C)}{|C|}\left(|C| -\frac{|C|}{|{M \backslash C}|}(|D|-|C|)\right) \leq NETV(D),
\end{equation}
\begin{equation}
NETV(C)\frac{{|C|}}{{|C|}}\left(\frac{|{M \backslash C}|-(|D|-|C|)}{|{M \backslash C}|}\right) \leq NETV(D),
\end{equation}
using $|{M \backslash C}|-(|D|-|C|)$ $=|M| - |C|-|D|+|C|$ $=|M \backslash D|$, we finally have
\begin{equation} \label{eq: theorem finally}
NETV(C) \leq \frac{|{M \backslash C}|}{|{M \backslash D}|} NETV(D).
\end{equation}
\end{proof}

 \begin{corollary} \label{cl: locally minimal}  Plugging Eq. \eqref{eq: indictr NETV} to Eq. \eqref{eq: theorem finally} we obtain

\begin{equation} 
per(C) = NETV(C) \leq \frac{|{M \backslash C}|}{|{M \backslash D}|} NETV(D) = \frac{|{M \backslash C}|}{|{M \backslash D}|} per(D) \,\, \forall D \supset C,
\end{equation}
thus any eigenset $C$, subset of a closed $M$, is a locally minimal perimeter set. 
\end{corollary}

If $|M|=\infty$, then $C$ is a minimal perimeter set. This happens since locally minimal perimeter subsets of surfaces with $|M| =\infty$ automatically become minimal perimeter, by $\frac{|M \backslash C|}{|M \backslash D|}=1$.
Furthermore, if $M$ is Euclidean (with $|M|=\infty$), then this corollary shows that $C$ is convex - which is a well known property for Euclidean $TV$. Hence, the new notion of locally minimal perimeter is a generalization of set convexity, with regard to total-variation theory.


\begin{figure} [htbp]
\includegraphics[width=0.5\textwidth]{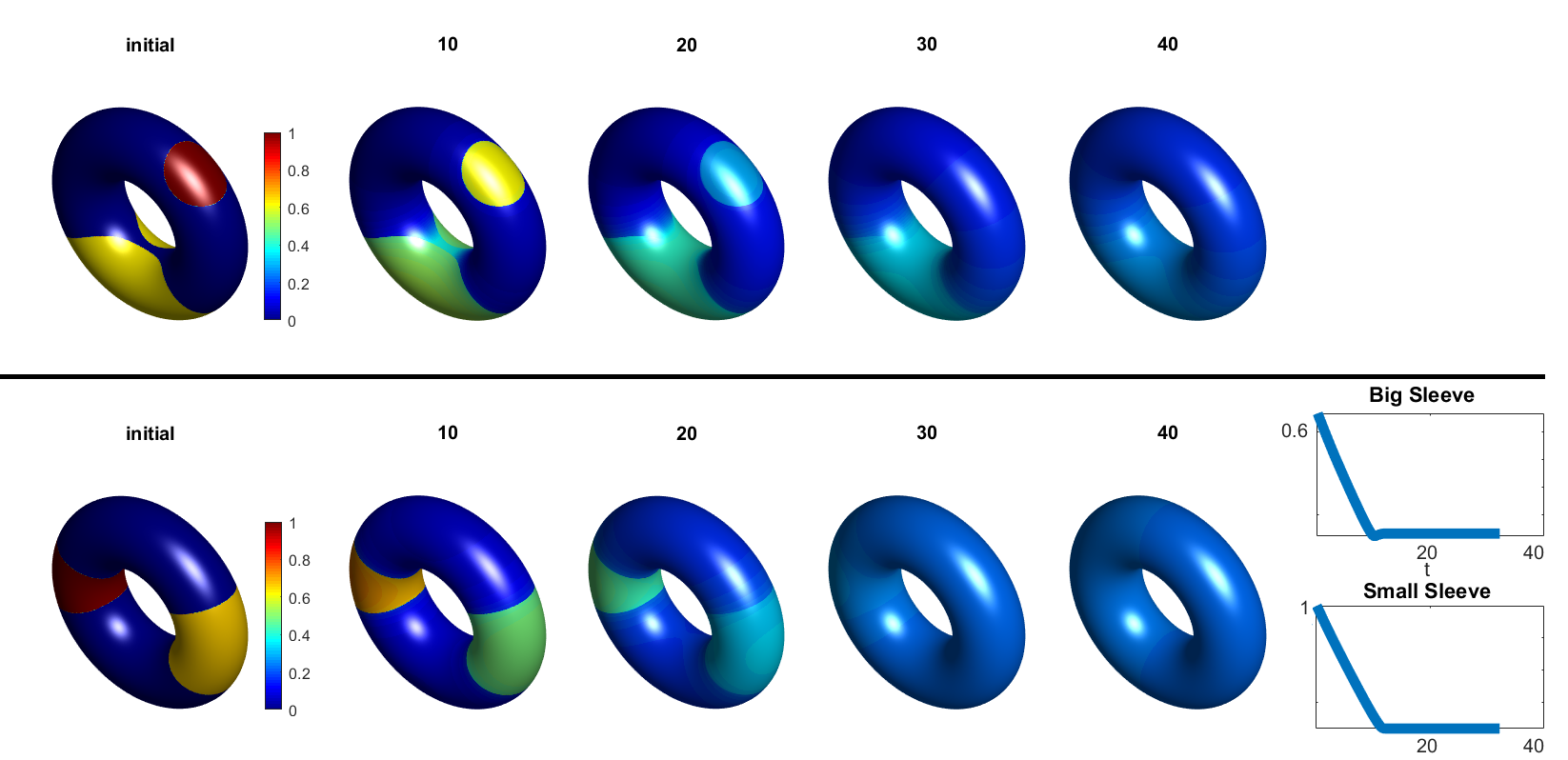}
\caption{Weighted indicator functions on Torus and their $NETV$ flow. First row: Two geodesical disks. The small disk is defined in "Euclidean position", i.e. results for Euclidean case carry over. Namely - the small disk is both geodesically convex, and a (locally) minimal perimeter set. The large disk is not in Euclidean position, and Euclidean laws do not carry over: It is not geodesically convex, nor is it of locally minimal perimeter. Throughout the flow it becomes a sleeve-like set, which is a minimal perimeter set, after which it decays completely. Bottom row: Indicator function of two sleeve sets $C_1, C_2$, weighted by the ratio of their areas: $f = \tilde{\chi}^{C_1} + \frac{|C_2|}{|C_1|}\tilde{\chi}^{C_2}$. The sleeve sets are eigensets, as shown in the appendix. By Eq. \eqref{eq: eigenvalue perimeter area}, and since $per(C_1)=per(C_2)$,  the eigenvalue ratio should be inverse the area ratio i.e. $\frac{\lambda_2}{\lambda_1}=\frac{|C_1|}{|C_2|}$. Thus we can reformulate as $f = \tilde{\chi}^{C_1} + \frac{\lambda_1}{\lambda_2}\tilde{\chi}^{C_2}$. If this is true, then by Eq.  \eqref{eq: linear decay} we expect the sets to completely decay at the same time - and indeed they do.} 
\label{Toroid}
\end{figure}

\subsection{Overview: Application of the Theory}
Eigenfunctions of the non-Euclidean total variation functional are stable, linearly decaying, modes of the minimizing flow. They are expected to be seen throughout the flow, especially upon convergence \cite{gilboa2013spectral}, \cite{bungert2020asymptotic}. Fig.  \ref{fig: spheres flow} demonstrates such linear decay - indicating that the decaying function is an eigenfunction.  Such eigenfunctions divide the non-Euclidean manifold to subsets, called "eigensets". For the Euclidean setting it is well established that such eigensets must be convex  - giving interpretability and understanding of the expected behaviour of the total-variation minimizing flow, and its regularization properties in general \cite{andreu2001minimizing}, \cite{bellettini2002total}. Until now, a fitting generalization was not proven for the non-Euclidean setting.

It was recently proposed \cite{fumero2020nonlinear}, based on empirical evidence, that eigensets of the non-Euclidean total variation are geodesically convex. However, Fig. \ref{fig: spheres flow} demonstrates numerically a linearly decaying set, i.e. an eigenset, that is not geodesically convex.

In our theoretical derivations  we begin by finding the expected values of the eigensets. These enable us to find the eigenvalue associated to the eigenfunction. The eigenvalue readily gives us understanding of the flow behaviour - as it is equal to the rate of linear decay of the eigenfunction.  Fig. \ref{Toroid} demonstrates numerically that the rate of decay is as expected. Finally - we derive a necessary condition for a set to be an eigenset, which we call the "locally minimal perimeter" condition. This condition is validated to be a non-Euclidean generalization of the well known Euclidean set convexity property.  In Fig. \ref{sinc_set}, we observe the locally minimal perimeter criterion on the stable modes of the flow.  The influence of the locally minimal perimeter criterion can also be seen across our shape processing experiments which follow, most notably inf Figs. \ref{sinc}, \ref{armadil0 filtering},  \ref{fig:deformation_process}, \ref{fig:stretch_flow}, \ref{sinc_z}.

\section{Spectral Decompositions}
 \label{sec: zero hom spect}
 Here we show straightforward extensions of some of the observations done in \cite{burger2016spectral} to our settings.
Let  $X$  be a space of functions on a surface $M$ equipped with a metric $g$, as described in Eqs. \eqref{eq: patmeteric manifold}, \eqref{eq: metric}.  Let $p:X \rightarrow X $ be a zero-homogeneous  operator, i.e.
\begin{equation}\label{eq: 0hom}
p(\alpha f)=\textrm{sign}(\alpha)p(f),\quad \alpha \in \mathbb{R}, f \in X,
\end{equation} 
with $p(0)=0$. 
We examine the following flow:
\begin{equation} \label{eq: flow}
u_t = -p(u(t)), \quad u(0) = f \in X, \,\, t\geq 0, 
\end{equation} 
 where $u_t=\frac{\partial u} {\partial t}$. We assume the flow exists and that the solution is unique. Unlike \cite{burger2016spectral}, here $p$ maps functions on non-Euclidean domains. Nevertheless, the time domain, denoted by $t$, is Euclidean. Note that elements in $\partial NETV$ are zero-homogeneous, thus, $NETV$ flow (subgradient descent of the energy) is a zero-homogeneous flow.  
We also assume the second time-derivative of $u$ exists in the distributional sense almost everywhere and define $\phi:t\rightarrow X$  as
\begin{equation} \label{eq: phi}
\phi(t) = t \cdot u_{tt}.
\end{equation}
Let $\psi$ be an eigenfunction with respect to $p$ with a positive eigenvalue, i.e. $\exists \lambda \in (0,\Lambda<\infty) : p(\psi) = \lambda \psi$. Let $f=\psi$, then the solution of Eq. \eqref{eq: flow} is
\begin{equation} \label{eq: linear decay}
 u(0)=\psi \Rightarrow u(t) = \begin{cases}
 (1-\lambda t) \psi\, & t \leq \frac{1}{\lambda}\\
 0 \, & t>\frac{1}{\lambda}
 	\end{cases}.
  \end{equation}
This can be verified by having, for $t < \frac{1}{\lambda}$, the relation 
\begin{equation}\label{eq:eig_decay}
p(u(t))=p((1-\lambda t) \psi)=p(\psi)=\lambda \psi = -u_t,
 \end{equation}
 where the second equality uses 0-homogeneity of $p$, the third equality uses the eigenfunction property, and the last equality is an evaluation of $u_t$ by \eqref{eq: linear decay}. For $t = \frac{1}{\lambda}$ we have a steady state since $p(0)=0$.  By uniqueness of the solution we are done.

 We note that since we are examining smoothing processes, $p$ in general is a positive semidefinite operator, $\langle f,p(f)\rangle \ge 0$, $\forall f \in X$. Thus the eigenvalues are positive. In the case of negative eigenvalues, the flow diverges (but for a finite stopping time can still have a solution).

  Thus, for eigenfunctions of positive eigenvalues  we get $\phi(t)=\delta(t-\frac{1}{\lambda})f$, i.e.  $\phi$'s energy is concentrated in a single scale (``frequency'') which corresponds to the eigenvalue of $f$,  $\lambda = \frac{1}{t}$. For a general $f\in X$, this motivates the interpretation of $\phi$ as a spectral transform of $f$, where the spectral components are positive eigenfunctions of $p$, in a similar manner to \cite{Gilboa_spectv_SIAM_2014,burger2016spectral,bungert2019nonlinear}.

We can compute the reconstruction formula, for a general stopping time $T$, using integration by parts (and assuming $u_t(0)$ is bounded), by
$\int_0^T \phi(t)\,dt = t u_t|_0^T - \int_0^T u_t\,dt$ $ =T u_t(T)-u(T)+f=-Tp(u(T))-u(T)+f.
$
Denoting the residual $R = T p(u((T)) + u(T)$, the following reconstruction identity holds $f=\int_0^T \phi(t)\,dt + R$. This gives rise to a filtered reconstruction via a filter $H:t\rightarrow \R$ as follows:
\begin{equation}\label{eq: nonlinear filtering}
f^{filtered}=\int_0^T H(t) \phi(t)\,dt + R.
\end{equation}

An important special case is the operator $p \in \partial NETV$, which is zero-homogeneous. Hence a $NETV$-minimizing flow can be used to filter signals by \eqref{eq: nonlinear filtering}. Moreover, in the case of an initialization with a single eigenfunction we expect a linear decay.  In Figs. \ref{Toroid}, \ref{sinc} we use these properties to demonstrate some of our theoretical findings from Sec. \ref{sec: thry}. In the following we demonstrate how this can be carried over to shape processing.





\begin{figure*} [htbp]
\includegraphics[width=\textwidth]{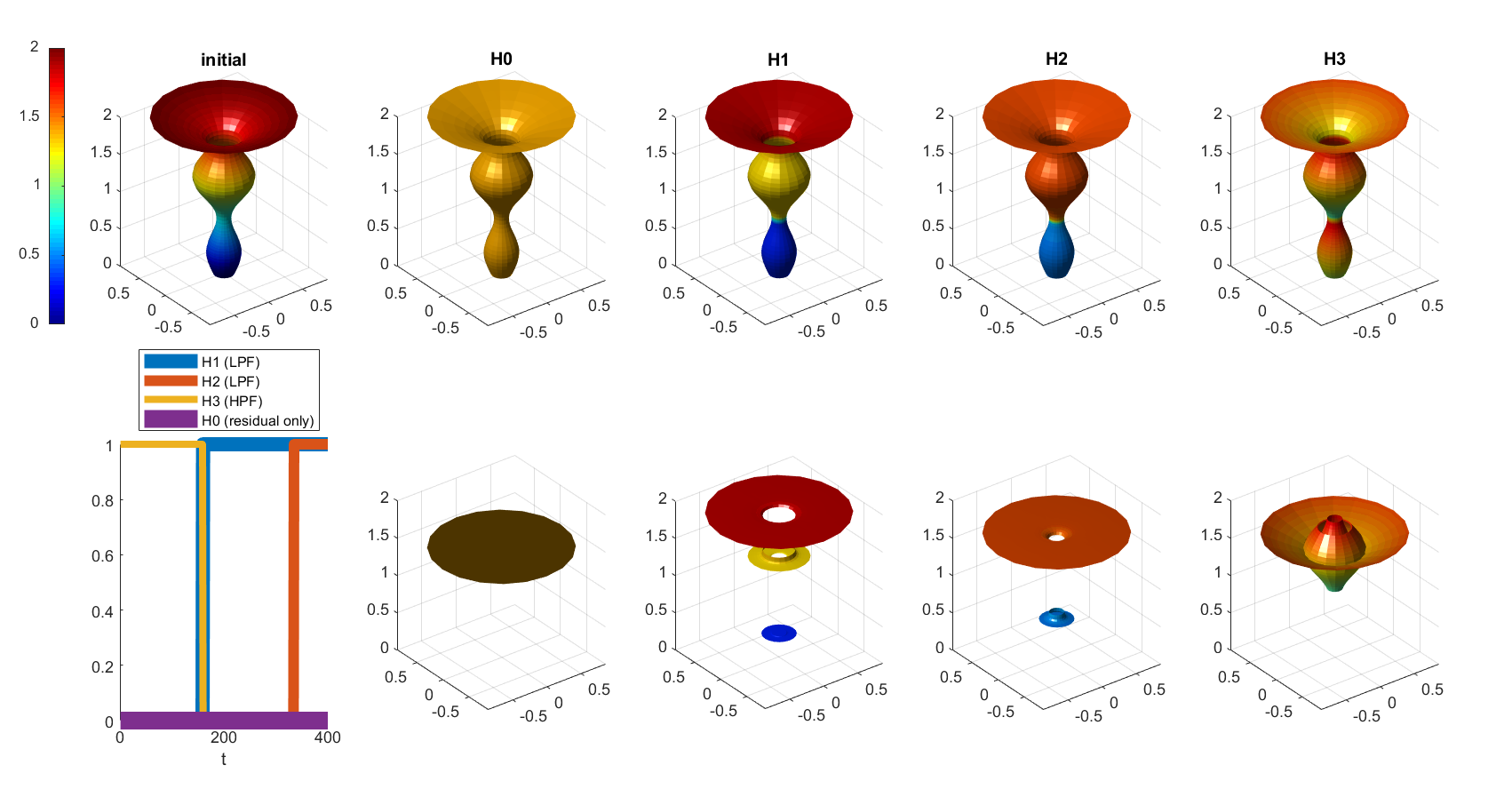}
\caption{Upper row: $M$ is the surface seen also in Fig. \ref{sinc_set}. $f$ is initialized as the z coordinate function, shown in color on the shape. A zero-homogeneous operator, $\partial NETV$, induces the spectral components - used for filtered reconstruction. Middle row depicts filtered coordinate function as color on shape. Bottom row: Having a geometrical meaning - the filtered coordinate function can replace the shape's coordinate -  Bottom row transforms middle row in this manner, followed by a simple graph cut for clearer visibility of vertex location. The  LPF (low pass filter) $H1$ exposes an underlying geometrical structure of three planes - corresponding to three "sleeve sets". The sleeve sets decay linearly throughout the flow, proving to be spectral components of $\partial NETV$. The geometrical meaning of linear decay is constant velocity. $H_2$ further filters and shrinks the upper two planes to one - showing that the bottom plane is of lowest eigenvalue, as it is the slowest decaying plane. $H_0$ shrinks the three underlying planes to a single one, the residual plane. $H_3$ adds the non-planar structures to the residual plane. The naive choice of coordinate function is not always natural to the processed shape.} 
\label{sinc}
\end{figure*}

\begin{figure} [htbp]
\includegraphics[width=0.5\textwidth]{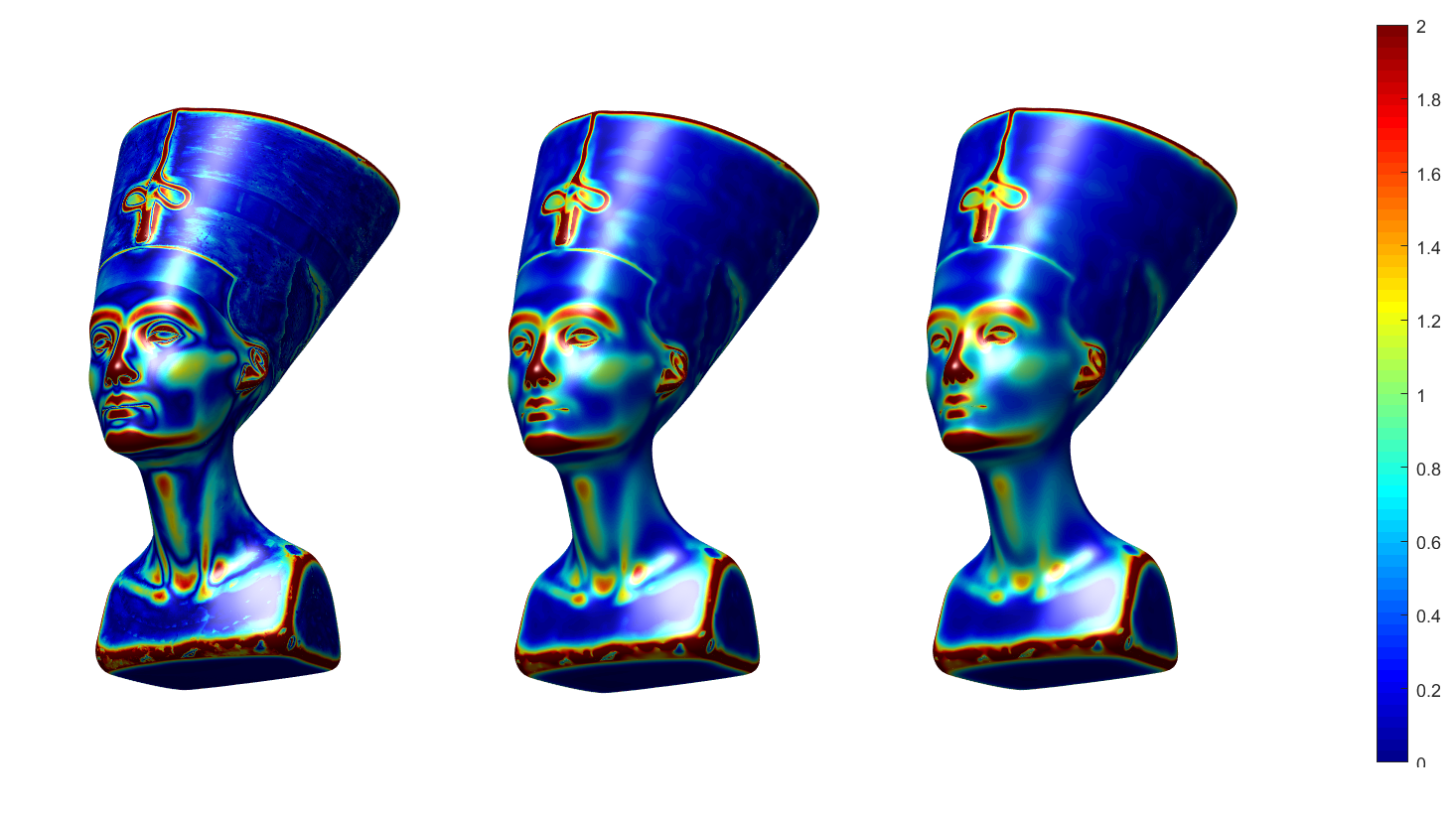}
\caption{Handling real world shapes, demonstrated on scan of the famous archeological statue "Bust of Queen Nefertity". In this example, geometric details are captured as a function that represents perturbation from an underlying smooth shape (further details in Sec. \ref{resmoothing method}). In color we see these geometric details experiencing a $NETV$ minimizing flow. To obtain the result of Fig. \ref{fig: nefertities}, we follow a procedure similar to the one presented in Fig. \ref{sinc}: Spectral representations are obtained from the flow and used for filtering. Followingly the original geometrical details are replaced by their filtered versions, which results with the final filtered shapes.}
\label{nefertity_scalar}
\end{figure}

\section{Shape Processing} \label{Methods}
 We suggest three methods for nonlinear filtering of shapes, in the framework of Sec.  \ref{sec: zero hom spect}. The methods differ by the choice of the operator $p$ of the respective flow. Each method is inspired by a different flow: $M1$ by the Heat Flow, $M2$ by cMCF and $M3$ by MCF. $M3$ uses $p \in \partial NETV$, hence theory of Sec. \ref{sec: thry} applies. $M1$ and $M2$ use different yet related operators, and theory regarding these is left as future work. Nevertheless - in all three methods we attain good feature control via manipulation of the spectral components. We will also demonstrate how the choice of different operators $p$ induces different qualities.

 So far, we processed a function $f$ on a 2D manifold $M$ embedded in 3D, via Eqs. \eqref{eq: flow}, \eqref{eq: phi}, \eqref{eq: linear decay}, \eqref{eq: nonlinear filtering}. Here we wish to process the manifold itself - and for this purpose we will choose a function $f$ that describes $M$. In the setting of parameterized surfaces, the surface function may be the function of choice, i.e. $f=S$, a vectorial function with three channels as the three coordinate functions $x_0(\omega_1,\omega_2), y_0(\omega_1,\omega_2), z_0(\omega_1,\omega_2)$. Note that $S$ also induces the intrinsic metric $g$. This choice is widely used for shape flows, thus it enables a comparison between our framework and the classical ones. Other representations can be used, see for example Fig. \ref{nefertity_scalar}. Denote the evolving shape at time $t$ as $S(t)$, and denote $c(t)$ as any evolving coordinate function of $S(t)$, that is $c(t)$ may assume $x(t)$, $y(t)$, or $z(t)$.

\textbf{}


\subsection{ Modifying flows for nonlinear spectral processing}
Our framework requires a zero-homogeneous flow evolving on a fixed metric, which is required to induce spectral linear decay of the eigenfunctions. Denote the metric induced by the initial shape $g_0$. It is fixed throughout the flow. Denote the evolving shape's metric $g_t$, which changes throughout the flow, i.e. it is not fixed.


Examining Heat Flow, MCF and cMCF we find that none of these flows is zero-homogeneous, and Heat Flow is the only one performed on a fixed metric. Hence, adaptations of these flows are required.

\subsection{Naive method: Unpaired Coordinate Spectral TV}


The naive approach utilizes a modification of Heat Flow for our framework. Heat Flow processes each coordinate function independently via Eq. \eqref{eq: lin diff}, utilizing the Laplace-Beltrami on the fixed metric $g_0$ throughout the flow. Thus it satisfies a fixed metric, but it is not zero-homogeneous, and a modification is required. By replacing the Laplace-Beltrami with the 1-Laplace-Beltrami of Eq. \eqref{eq: g p Laplacian} zero-homogeneity is achieved, which results in the operator
$-p_{Naive}(c) := \Delta_{g_0,1}c$,
and a per-coordinate flow is defined by setting $p(u(t))=-p_{Naive}(c(t))$ in Eq. \eqref{eq: flow}. Each channel evolves separately, hence the name "unpaired coordinates". We can now perform nonlinear spectral filtering as in Eq. \eqref{eq: nonlinear filtering}, demonstrated on the meteor model
in Fig. \ref{meteor}. Note the axis squaring effect, which violates rotation invariance, and also restricts the underlying structure, resulting with bad separation between structure and detail. With that said, it does provide relatively aesthetic results, if one desires such squaring effect.

\subsection{Method 1 (M1): Shape Spectral TV} \label{VTV method}

\begin{figure} [!htbp]
\includegraphics[width=0.5\textwidth]{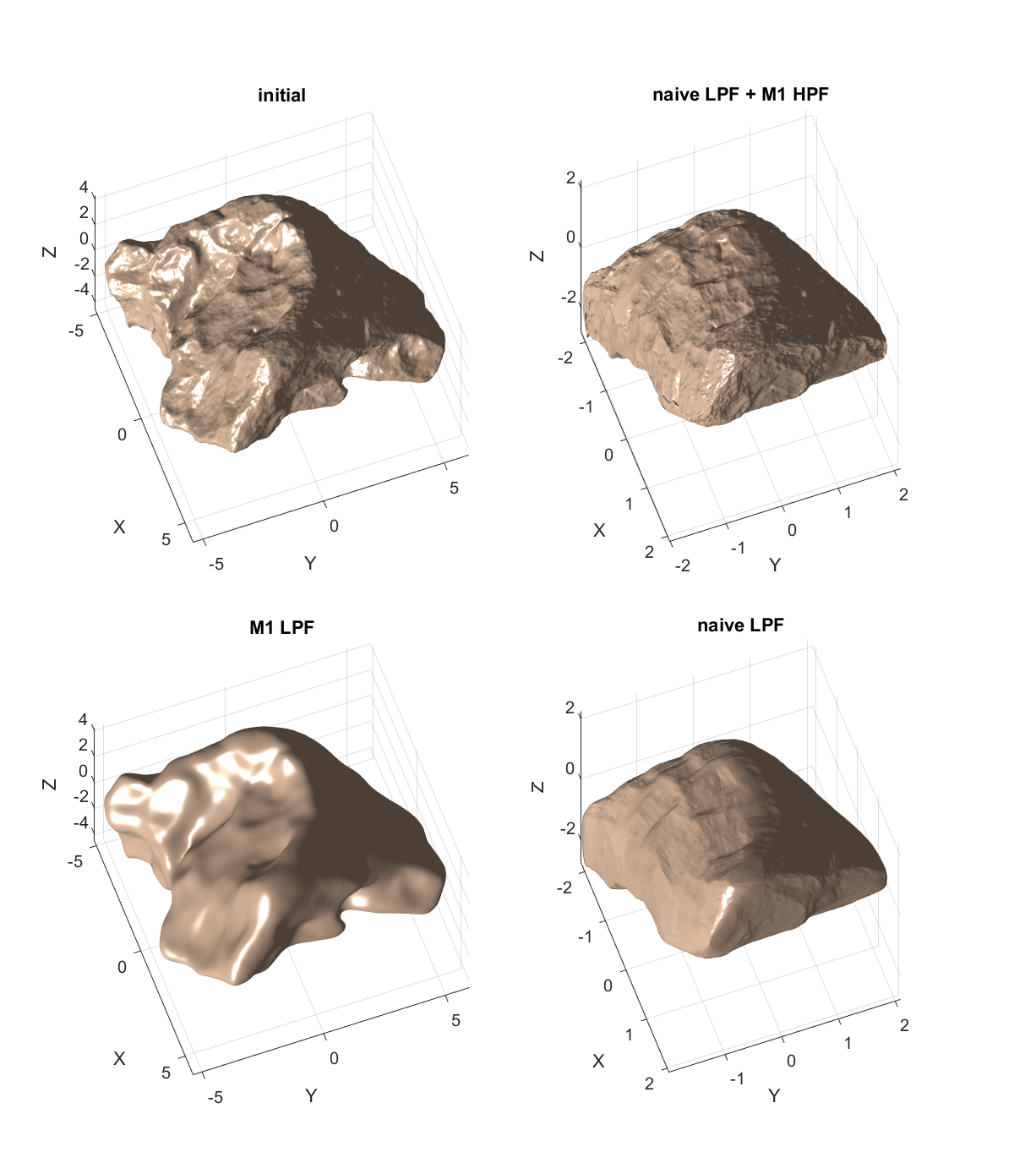}
\caption{A real-world 3D scan of a Meteorite\protect\footnotemark. In this example, the shape is described by its coordinate functions. Unlike linear shape processing, the nonlinear case requires special attention to deal with the correlation between the $X,Y$ and $Z$ coordinate functions. In Method 1 (M1) such considerations are accounted for. From top left, counter clockwise: 1) original model; 2) M1 low-pass filter (LPF); 3) Naive method low-pass filtering. ignoring the inter-coordinate correlation, this method has an axis-squaring effect. Even such harsh filtering with the naive method dose not yet fully smooth out details, suggesting lack of feature control. Contrary, M1 smooths out details without distorting underlying structure, as expected from a smoothing procesdure; 4) Details obtained via $M1$ high-pass filter (HPF) are added to the model after squaring it with the naive method LPF} 
\label{meteor}
\end{figure}
\begin{figure} [!htbp]
\includegraphics[width=0.5\textwidth]{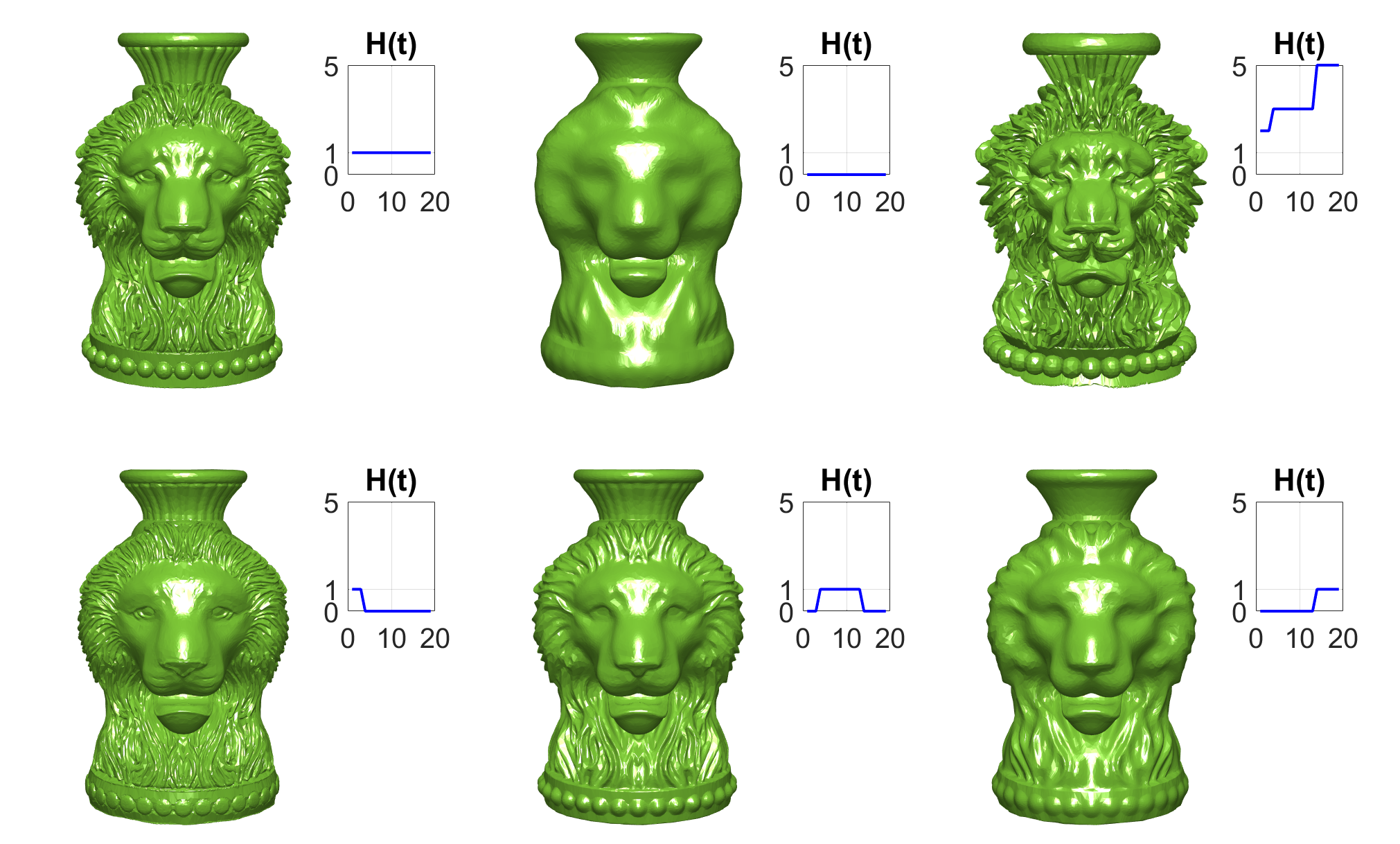}
\caption{$M1$ filtering. Top row, left to right: All-pass reconstruction; residual; 3-bands amplification. Bottom row: each of the 3 bands isolated on top of the residual. As $t$ grows, choosing $H(t)$ greater or less than 1, results in amplification or attenuation of courser details, as can be observed e.g. in the model's contour and clumping of hair-strands.} 
\label{lion vase bands}
\end{figure}

Here we take into account shape coordinate inter-correlations, i.e. we go from coordinate to shape processing. We apply a vectorial flow (as in VTV) on meshes, which results in the operator, 
\begin{equation}\label{eq:m1_operator}
	-p_{M1}(c) := \nabla_{g_0} \cdot \left(\frac{\nabla_{g_0}c}{\sqrt{\sum_{\tilde{c}=x,y,z}|\nabla_{g_0}\tilde{c}|^2}}\right).
\end{equation}
Note that the metric is fixed as $g_0$. We can also verify that the operator is zero-homogeneous.

Remark: Similar flows where proposed in the past. One important example is \cite{elmoataz2008nonlocal}, from which we borrow the combined gradient magnitude of the denominator - designed to account for the inter-correlation of the coordinate function. With that said, the gradient and divergence operators used in \cite{elmoataz2008nonlocal} are significantly different - as they are obtained for general graph structures, while we use operators which account for the surface properties of the mesh. For instance, the graph gradient used in \cite{elmoataz2008nonlocal} is calculated per-vertex and has a non-fixed dimensionality that is equal to the number of edges connected to said vertex, while the surface mesh gradient we use is obtained per-triangle, and is parallel to said triangle, i.e. it is embedded in $\mathbb{R}^3$.

%
%
The flow is followed by per-coordinate spectral processing as in Eq. \eqref{eq: phi}, \eqref{eq: nonlinear filtering} - 
thus $x,y,z$ inter-correlation is preserved. Compared to the naive approach, $M1$ preserves better the initial underlying structure, as demonstrated in Fig. \ref{meteor}. Good multi-scale feature control is demonstrated as well in Fig. \ref{lion vase bands}.

\subsection{Method 2 (M2):  Conformalized 3-Laplace } \label{conformalized method} 

Here we modify cMCF \cite{kazhdan2012can} to our framework, by presenting a  conformalized $\mathcal{P}$-Laplace, as described below. Our flow inherits cMCF's limb-head smoothing capabilities (Fig. \ref{dynos}), which we then use for shape filtering. The metric of cMCF is $\tilde{g}_t=\sqrt{|g_0^{-1}g_t|}g_0$, is not fixed. The operator driving the flow, the conformalized Laplace-Beltrami, $ \sqrt{\frac{|g_0|}{|g_t|}}\nabla_{g_0} \cdot \nabla_{g_0}$, depends on the evolving shape's metric $g_t$. To achieve a fixed metric, we re-interpret $|g_t|$ as an operator on the fixed metric $g_0$. This is valid since the diffused shape defines both the diffused function as well as the evolving metric. This affects homogeneity, as shown below. We define the conformalized $\mathcal{P}$-Laplace as,
\begin{equation}\tilde{\Delta}_{g,\mathcal{P}}(c): = \sqrt{\frac{|g_0|}{|g_t|}}\nabla_{g_0} \cdot (|\nabla_{g_0} c|^{\mathcal{P}-2}\nabla_{g_0} c ).
\end{equation}
By Eq. \eqref{eq: metric} we have that $|g_t|$ is absolutely 4-homogeneous, hence $\tilde{\Delta}_{g,\mathcal{P}}$ is $\mathcal{P}-3$ homogeneous,
$\tilde{\Delta}_{g,\mathcal{P}}(\alpha c)$ 
$ = \sqrt{\frac{|g_0|}{|\alpha|^4|g_t|}}\nabla_{g_0} \cdot (|\alpha|^{\mathcal{P}-2} |\nabla_{g_0} c|^{\mathcal{P}-2}) \alpha \nabla_{g_0} c$
 $= \frac{\alpha }{|\alpha |^{4-\mathcal{P}}}\tilde{\Delta}_{g,\mathcal{P}}(c).
$
Thus we choose $\tilde{\Delta}_{g,3}$  as a zero-homogeneous modification of the conformalized Laplace. 
Once again inter-correlations are accounted for, as in \cite{elmoataz2008nonlocal}, yielding the operator,
\begin{equation}
-p_{M2}(c) := \sqrt{\frac{|g_0|}{|g_t|}}\nabla_{g_0} \cdot \left(\sqrt{\sum_{\tilde{c}=x,y,z}|\nabla_{g_0} \tilde{c}|^2}\nabla_{g_0} c \right)
\end{equation}
\begin{figure}[!htbp]
\includegraphics[width=0.5\textwidth]{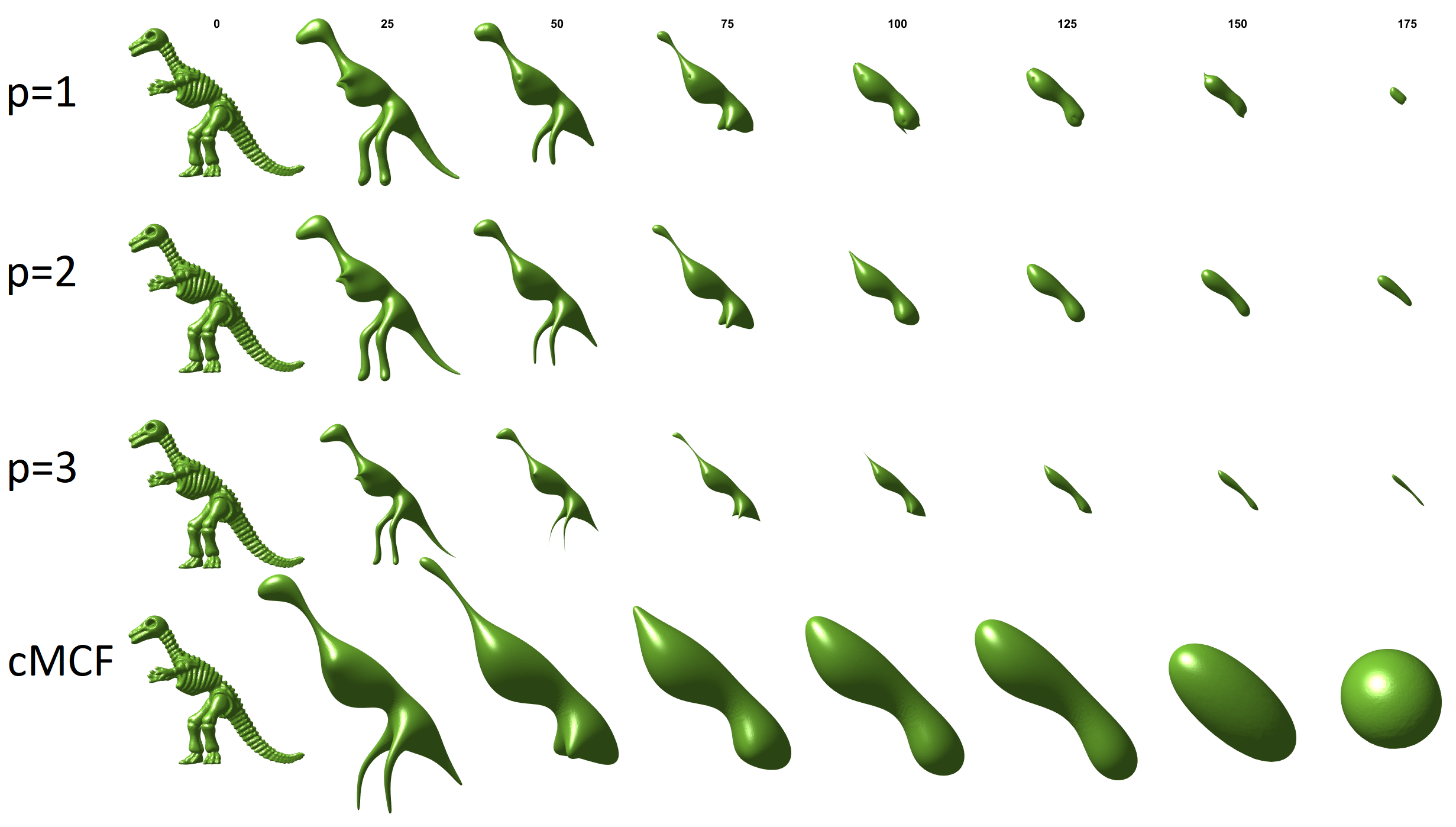}    
\caption{Three upper rows: M2 conformalized $\mathcal{P}$-Laplace flows. $\mathcal{P}=2$ is an unscaled version of the conformalized Mean Curvature Flow (cMCF). $\mathcal{P}=1$ is a new conformalized shape TV flow. For $\mathcal{P}=3$ the flow is zero-homogeneous. Bottom row: cMCF from \cite{kazhdan2012can}
} \label{dynos}
\end{figure}
The flow is followed by nonlinear spectral filtering, Eq. \eqref{eq: nonlinear filtering}. Editing extremities, a capability inherited from our conformal 3-Laplace flow, is demonstrated in Fig. \ref{victoria michael limbs}, where extremities are in the form of human limbs and head.

\begin{figure}[!htbp]
\includegraphics[width=0.5\textwidth]{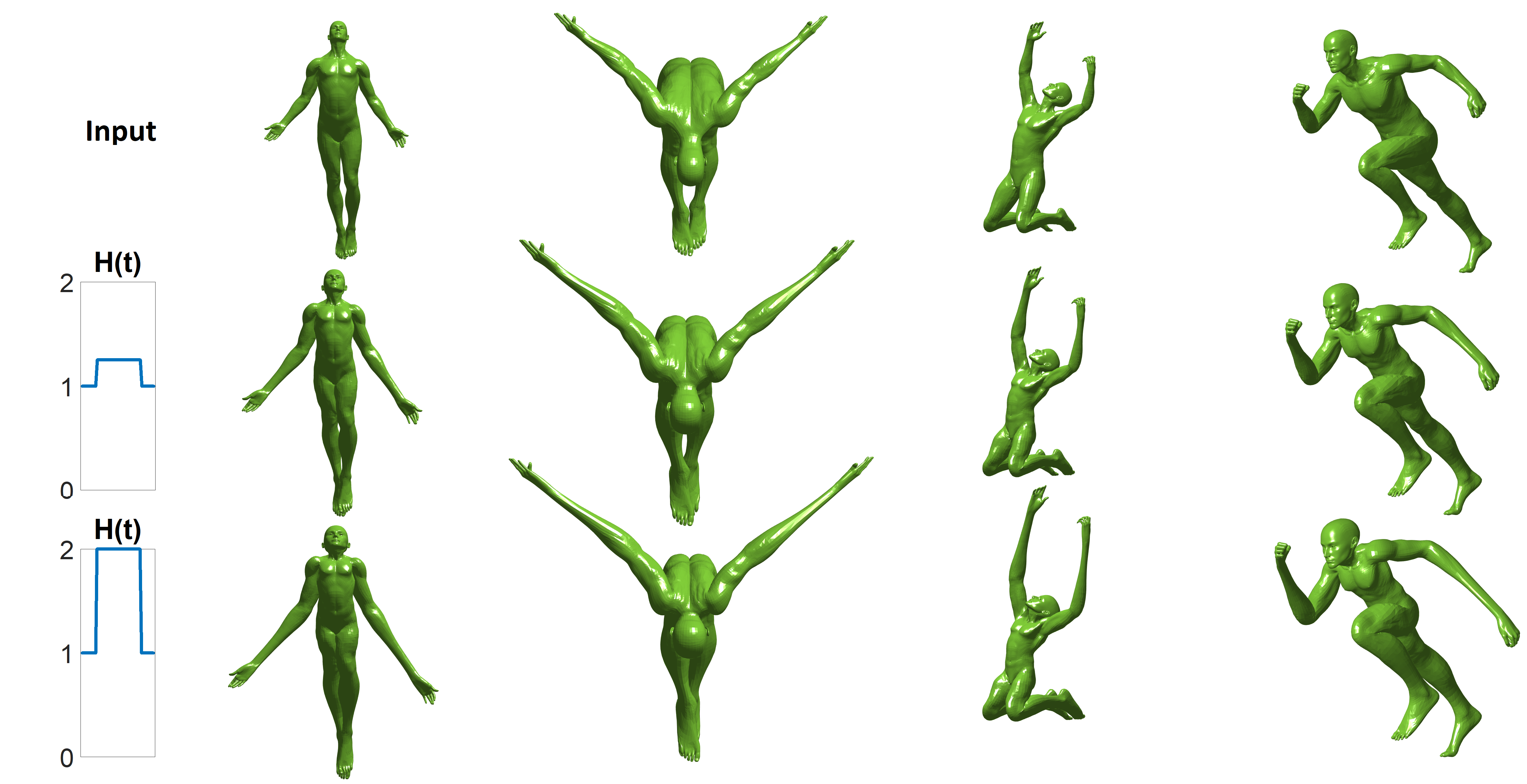}
\caption{Shape exaggeration by M2 spectral filtering. 
The filter inherits its properties from the conformalized 3-Laplace flow (Fig. \ref{dynos}), which translates to interesting limb-head editing. Isometry robustness is demonstrated as well. }\label{victoria michael limbs}
\end{figure}

\subsection{Method 3 (M3):  Directional Shape TV} \label{resmoothing method}
Mesh TV smoothing typically preserve pointy surface points, e.g. tip of chin \cite{fumero2020nonlinear} or ears \cite{elmoataz2008nonlocal}. Here we propose a method that preserves edges, e.g. muscle contour, similarly to TV processing of images. While M1 and M2 utilized modifications of Heat Flow and cMCF, M3 draws inspiration from MCF.

\footnotetext{\href{https://sketchfab.com/3d-models/cmu-edmore-meteorite-c20db897b4e14ef9a1cfe794e22037cc}{Edmore Meteorite 3D scan}, courtesy of Prof. \href{https://www.davis-sculpture.com/}{Jeremy Davis}, Central Michigan University, Department of Art and Design}

MCF already has a thoroughly researched fixed-metric zero homogeneous modification: The TV flow as applied to gray-scale images \cite{kimmel2000images}. 
For a surface represented as $S=(x,y,f(x,y))$, this modification entails constraining the evolved shape to be of the form $S(t)=(x,y,f(x,y,t))$ . This is enforced by constraining each point on the surface to evolve in direction $\hat{z}$ (perpendicular to the $x,y$ plane). We note that unconstrained MCF would necessarily violate this form of $S(t)$, as it theoretically converges to a singular point.

Our third  method aims to generalize the above direction-constraint to general shapes, hence the name "directional". The $x,y$ domain is generalized to be an over-smoothed version of the initial shape which we denote $\hat{S}$. Each $p \in S$ is mapped to a $\hat{p} \in \hat{S}$. The direction of evolution is fixed as $\hat{d}=\alpha\frac{S-\hat{S}}{|S-\hat{S}|}$, where $\alpha$ is a sign indicator which ensures $\hat{d}$ points "outwards".
Finally, the evolving initial surface  is represented as $f=\alpha|S-\hat{S}|$. Note that $S = \hat{S} + f\hat{d}$.
This method is a generalization in the following sense: Consider the form $S=(x,y,f(x,y))$, choosing $\hat{S}=(x,y,0)$, we have that $\hat{d}=\hat{z}$, and $\alpha|S-\hat{S}|=(0, 0, f(x,y))$.


The notion processing $d$ 
This method belongs to a vast family of shape processing techniques, operating on the displacement field



We advocate the choice of $\hat{S}$ as a cMCF smoothed version of $S$, since cMCF was shown to provide a conformal mapping from $S$ to $\hat{S}$.
By construction - the metric is fixed and inter-correlations are accounted for. The proposed zero-homogeneous operator (acting on a scalar-valued function $u$) is,
\begin{equation}
-p_{M3}(u): =  \nabla_{g_0}\cdot \frac{\nabla_{g_0}u}{|\nabla_{g_0}u|},
\end{equation}
where $u(t)$ is the evolution of $f$ at time t, which results in $\frac{\partial S}{\partial t} = \nabla_{g_0}\cdot \frac{\nabla_{g_0}u}{|\nabla_{g_0}u|}\hat{d}$, satisfying the imposed directionality. Finally, $f$ is filtered as in Eq. \eqref{eq: nonlinear filtering}, and a filtered shape is obtained by $\hat{S} + f^{filtered}\hat{d}$.
Being closely related to spectral TV on images, this method preserves detail well, as demonstrated in Fig. \ref{armadil0 filtering}. Though inspired by MCF, the proposed flow is substantially different.

\begin{figure} [!htbp]
\includegraphics[width=0.5\textwidth]{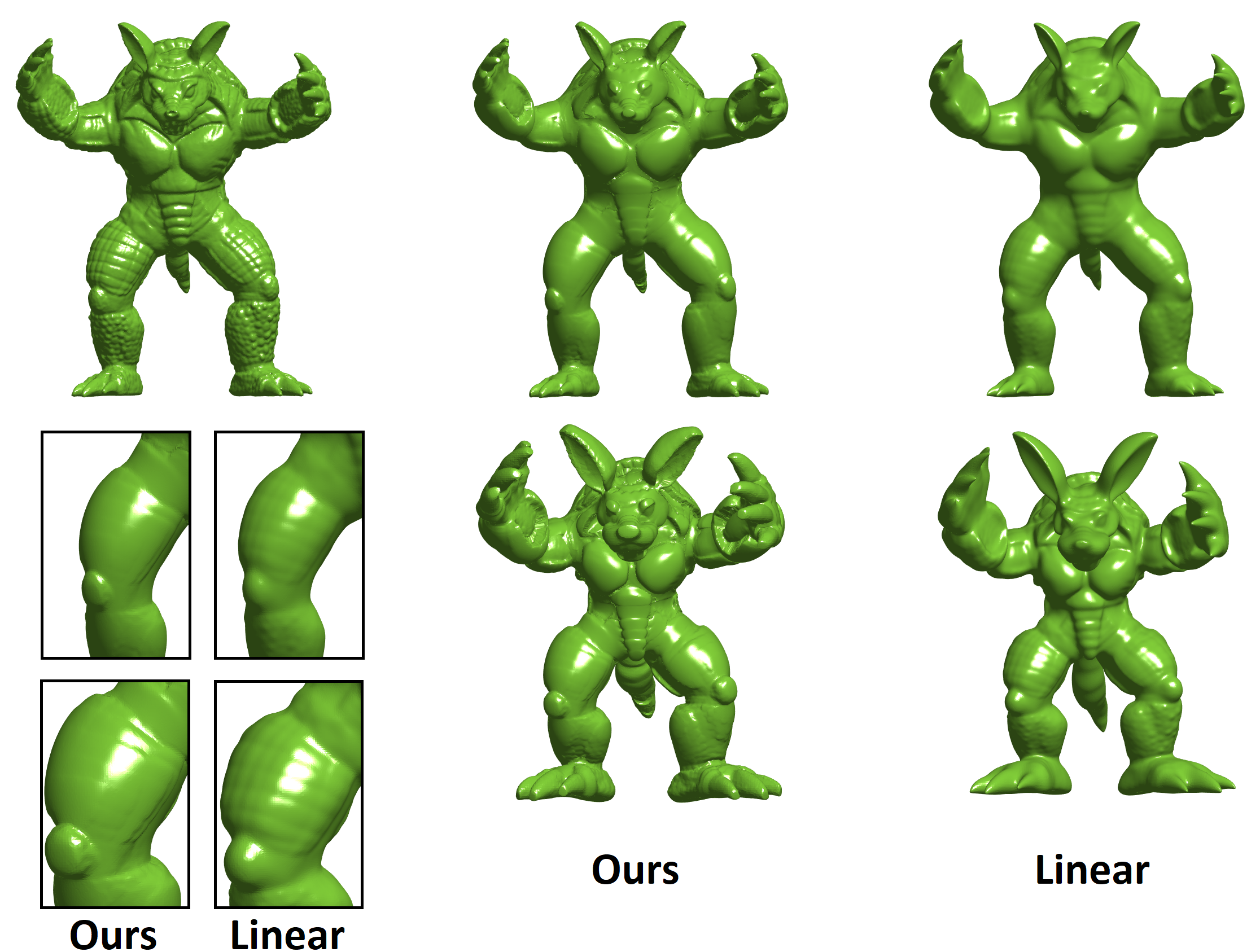}
\caption{Upper row: Low-pass filtering applied for smoothing. All our methods can utilize rough time discretization for runtime efficency at the expence of reconstruction error. In this example, M3 is used, and completed in approximately 30 seconds. In contrast, the Laplace-Beltrami method involves SVD for 2000 eigenvectors of a $[17\cdot 10^3\times 17\cdot 10^3]$ sparse matrix, requiring about 5 hours (2.5 orders of magnitude slower). Nonetheless, more efficient linear methods exist, such as in \cite{cignoni2005simple}, which take the same amount of time. Bottom row: Smooth caricaturization via bandpass exaggeration coupled with smoothing. We compare the exaggeration capabilities of M3 with those in \cite{cignoni2005simple}, which also exaggerates shapes using the deformation from an over-smoothed underlying shape. \\Bottom left (zoom-in): Our approach effectively smooths smaller details, such as the scales of the armadillo, while maintaining larger structures like the knees. Our theoretical findings correlate "resistance to smoothing" of geometrical structures with generalized convexity and the ratio of perimeter to area. The knees, being more convex-like with a lower perimeter to area ratio than the scales, indeed demonstrate greater resistance to smoothing. This results in superior separation of detail compared to both Laplace-Beltrami smoothing and the shape exaggeration method in \cite{cignoni2005simple}. For additional reference, exaggeration without smoothing is presented in the appendix, Fig.~\ref{fig:4_armadil_caricature}, where we also demonstrate consistency with the caricaturization principles discussed in \cite{sela2015computational}.}

\label{armadil0 filtering}
\end{figure}


\section{Shape Deformation}
Here we show a novel application. For the first time we perform total-variation for the shape deformation task. We will see that this induces a piecwise-constant deformation field, where the deformation concentrates on small-perimeter boundaries, showcasing the eigenset properties derived in Sec. \ref{sec: thry}.

The shape deformation problem involves finding a plausible transformation of a given shape while accommodating constraints. 
 The constraints are specified by the user as points, or regions of the shape, which are to be moved away from their original location. The deformation process should ensure that the resulting shape maintains its structural integrity while sufficing the constraints. 

 
To this end we propose a constrained minimization of the total-variation of the displacement field, defined as 
\begin{equation}\label{eq:deformation_field}
    d':=S'-S,
\end{equation}
where $S$ is the coordinate function of the original given shape, and $S'$ is the resulting shape. The minimization process is provided below as Algorithm \ref{alg:tv_deformation}. The algorithm uses a minimization process proposed by \cite{bronstein2016consistent} for minimzing L1 norms on manifolds. A slight difference is, that we use a vectorial version of \cite{bronstein2016consistent}. The minimizing operator of choice is $p_{M1}$ of Eq. \eqref{eq:m1_operator}. Minimization is performed while enforcing the deformation constraints.  As an initial solution to the process, we use the gradient-based linear deformation proposed in \cite{botsch2007linear}.  We use a matrix version of $p_{M1}$, denoted $P_{M1}^{matrix}$, which is updated at each iteration. Further details are available in Appendix \ref{app:shape_deformation}.

As illustrated in Fig. \ref{fig:deformation_process}, the displacement field $d'$ gradually becomes piecewise constant during the shape deformation process.  Interestingly, the deformation tends to concentrate on boundaries with small perimeters, suggesting a potential connection to eigenfunctions. We validate this hypothesis numerically in Fig. \ref{fig:stretch_flow} by analyzing $\|d'\|_g$,the magnitude of the field. See Fig. \ref{fig:snail}  for method comparison. For completeness, we added additional results in the appendix, Fig. \ref{fig:additional_stretch}.
 
\begin{algorithm}
\caption{Shape Deformation}\label{alg:tv_deformation}
\begin{algorithmic}
\State $S \gets \text{original uncostrained shape coordinate function}$
\State $b \gets \text{boundary indices}$
\State $\kappa \gets \text{boundary constraints}$
\State $w \gets \infty$ \Comment{ A high constraint weight}
\State $S' \gets \text{initial solution}$
\While{True}
        \State $S'_{cached} \gets S'$
        \State $\partial VNETV \gets \text{construct the minimizing operator considering }S, S'$
        \State $S' \gets \text{minimize($\| \partial VNETV(S'-S)\|_2^2 + w \cdot \|S'(b)-\kappa\|_2^2$}$)
        
        \If{$S'_{cached} \approx S'$} \Comment{Convergence stopping condition}
                \State break
        \EndIf
\EndWhile\\
\Return S'
\end{algorithmic}
\end{algorithm}


\begin{figure} [htbp]
\includegraphics[width=0.5\textwidth]{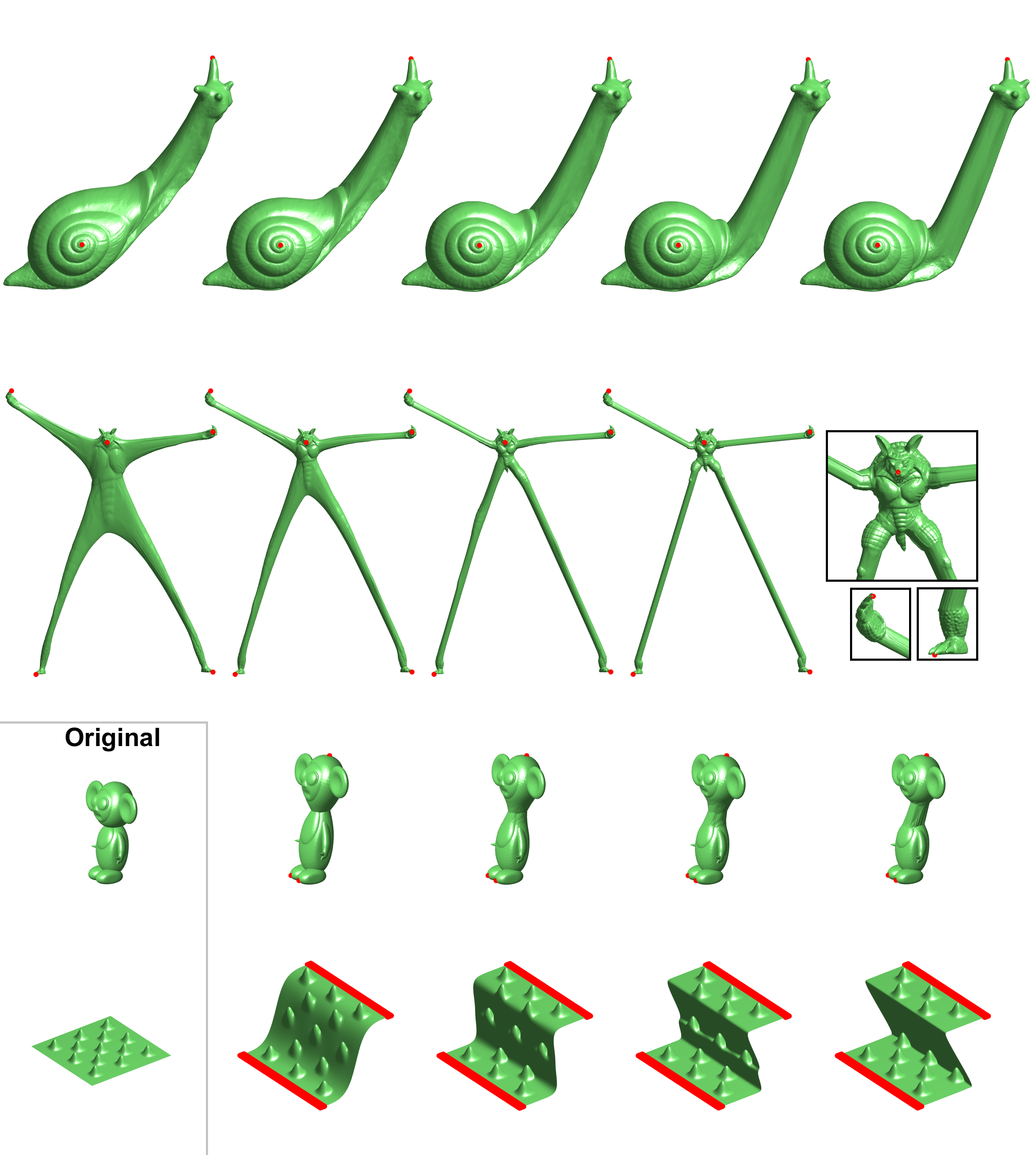}
\caption{Our TV regularization process of the displacement field in the shape deformation setting for Snail, Armadillo, Cheburashka and Knubble models. The deformation constraints, which are all visible here, are marked with red dots. All cases show the tendency for a piecewise constant deformation field. The sets of constant deformation inherit their properties from the eigensets of Sec. \ref{sec: thry}, namely - their boundaries have small perimeters, e.g., neck and elbows. Note in the zoomed-in windows (Second row, most-right),  how the not-stretched parts remain intact.}\label{fig:deformation_process}

\end{figure}


\begin{figure} [htbp]
\includegraphics[width=0.5\textwidth]{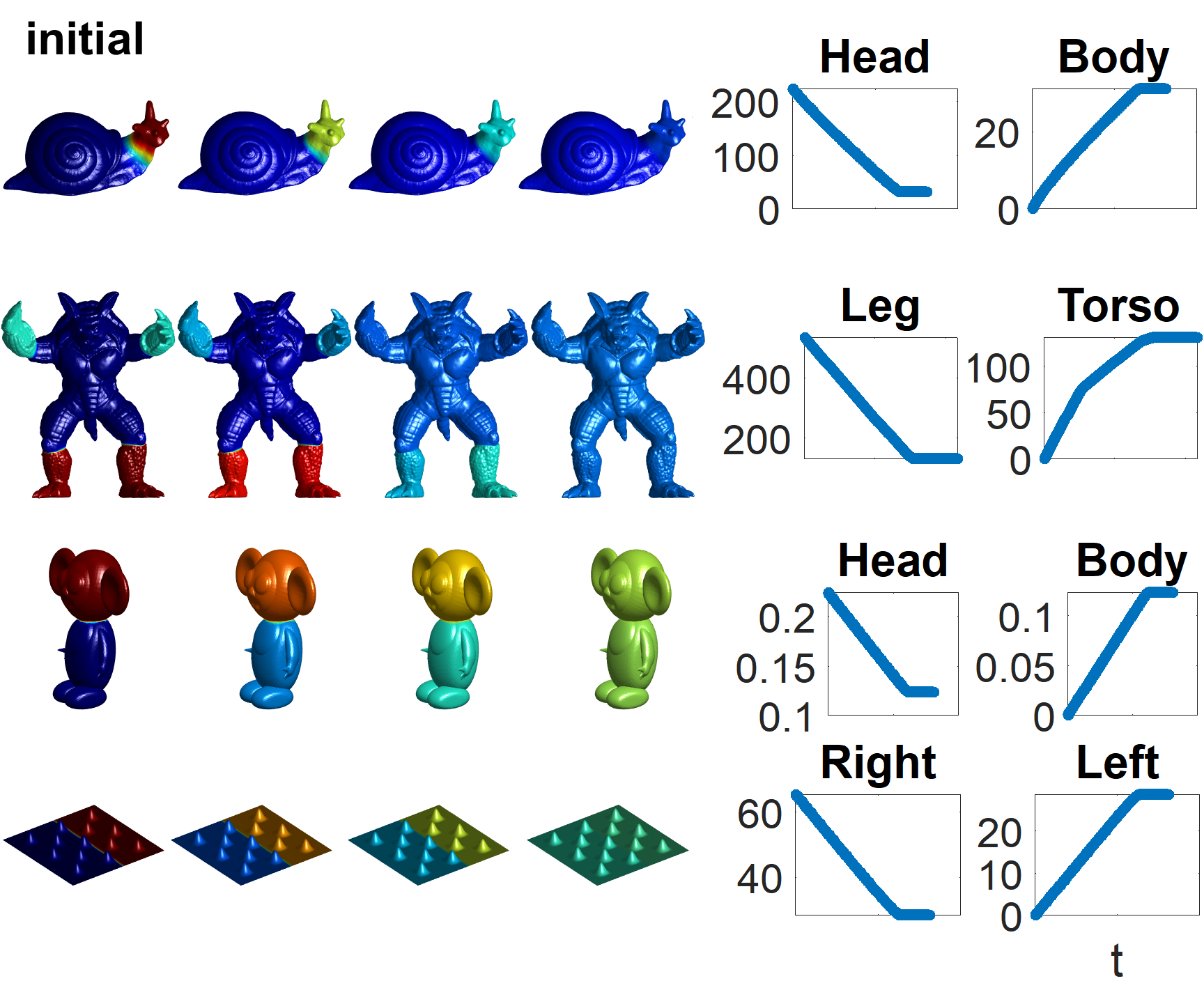}
\caption{ The magnitude of our resulting deformation field $d'$ tends to be constant on $NETV$ eigensets, as validated here numerically. Similarly to Figs. \ref{fig: spheres flow}, \ref{Toroid}, we check whether a set is an eigenset by testing whether it is a linearly decaying mode of the $NETV$ flow. Left Column: After obtaining the deformation (right column in Fig. \ref{fig:deformation_process}), we calculate the magnitude of the deformation field vectors, $\|d'\|_g$. We visualize these magnitudes using color, and it can be observed that they exhibit a piecewise-constant pattern.
Middle Columns: An $NETV$ minimizing flow is applied to the magnitudes $\|d'\|_g$. The sets of constant value remain intact until they completely decay. 
Right Columns: Here the numerical values of the sets are displayed throughout the flow. A linear decay can be observed - implying that these are indeed eigensets.
Observations: 
1) In the snail case, the initial $\|d'\|_g$ slightly deviates from the typical crisp boundaries observed in other shapes. Nonetheless, clearer boundaries appear early-on in the flow,  and the linear decay quickly becomes similar to the other shapes. 2) 
The decay of the Armadillo's torso is different since it exhibits two distinct rates. This can be explained by our theoretical findings: One rate happens before the arm sets coincide with the torso set, and another after. Following the coincidence of the arms and torso sets, the area increases while the perimeter decreases due to the disappearance of the arm perimeters. This leads to a reduction in the eigenvalue $\lambda$ according to Eq. \eqref{eq: eigenvalue perimeter area}, resulting in a decrease in the decay slope - as evident in the observed behavior and as anticipated by Eq. \eqref{eq:eig_decay}.
}
\label{fig:stretch_flow}
\end{figure}

\begin{figure} [htbp]
\includegraphics[width=0.5\textwidth]{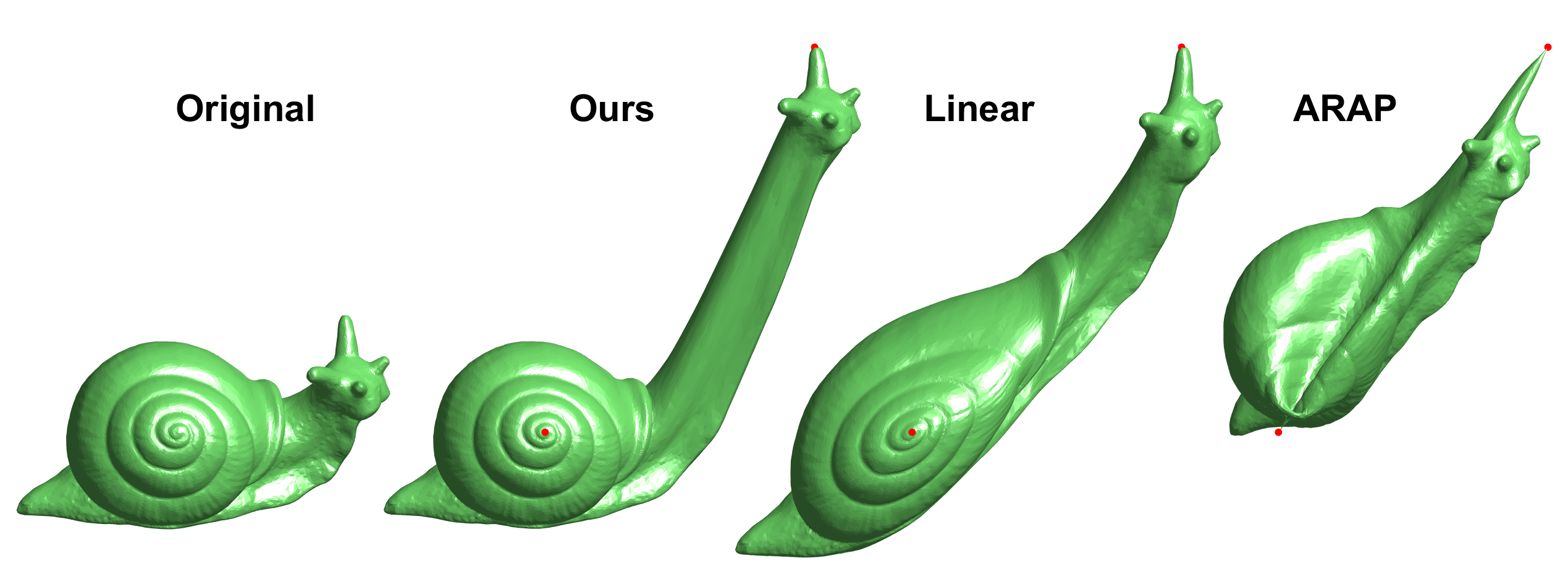}
\caption{A shape-deformation method comparison for the Snail\protect\footnotemark model, considering the same constraints in all methods. Upon using our method for shape deformation with total-variation, the stretch is concentrated on boundaries of plausible shape segments . In contrast -  linear methods stretch the shape globally, and As-Rigid-As-Possible methods (ARAP) globally penalize the local stretch. Other comparisons are available at Fig. \ref{fig:additional_stretch}.}\label{fig:snail}
\end{figure}

\section{Discussion}
\subsection{Non-differentiable shapes: A limitation of our theory}
Our theoretical findings rely on the notion of "good" metric spaces \cite{miranda2003functions}. In the context of our shape processing applications, these metrics are defined by the specific shapes being processed. However, within the domain of computer graphics, the shapes encountered can often exhibit high non-differentiability, which may even be enhanced by discretization. Consequently, our assumptions may not be suitable in such cases. While works concentrating on such considerations where conducted for the Laplacian operator (see for instance \cite{wardetzky2007discrete}, \cite{sharp2020laplacian}), this remains for future work regarding the operators used in this paper.

\begin{figure} [htbp]
\includegraphics[width=0.5\textwidth]{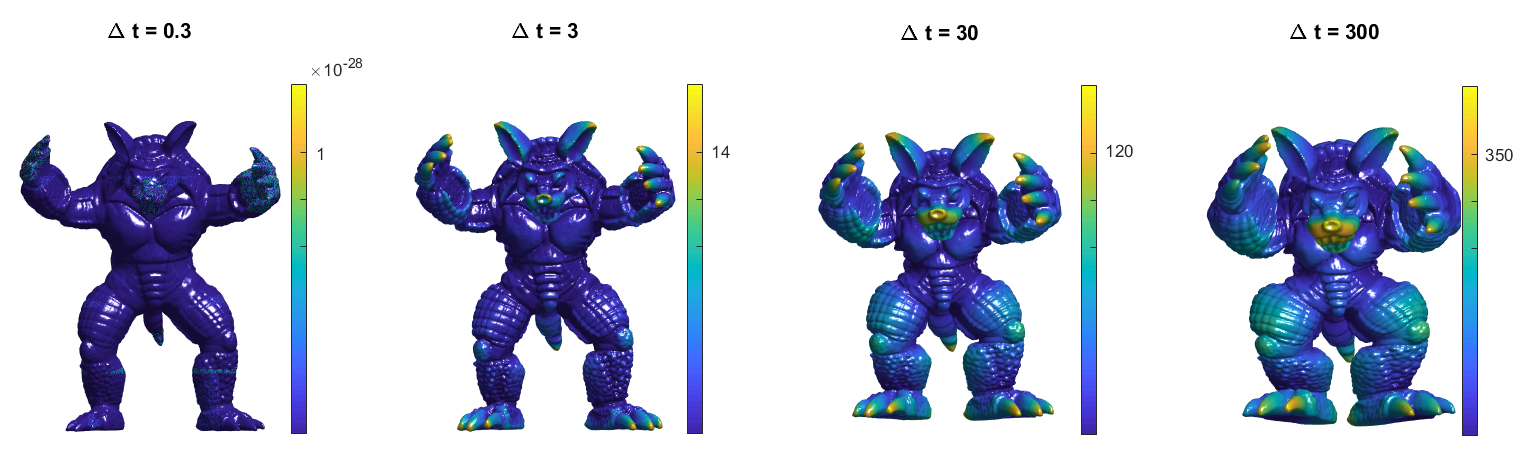}

\caption{Time step size of the discrete flow is an important parameter, controlling a trade off between flow accuracy, and execution complexity. An all-pass filter should perform reconstruction (in all our methods), however too large time steps result with inaccurate flow, which yields inaccurate spectral representations and reconstruction. Here we see all-pass filters of the Armadillo model using increasingly bigger time steps ($M1$). Reconstruction error is shown in color per-vertex, and is larger with time step size, as expected.}
\label{fig: rec err}
\end{figure}

\begin{figure} [htbp]
\includegraphics[width=0.45\textwidth]{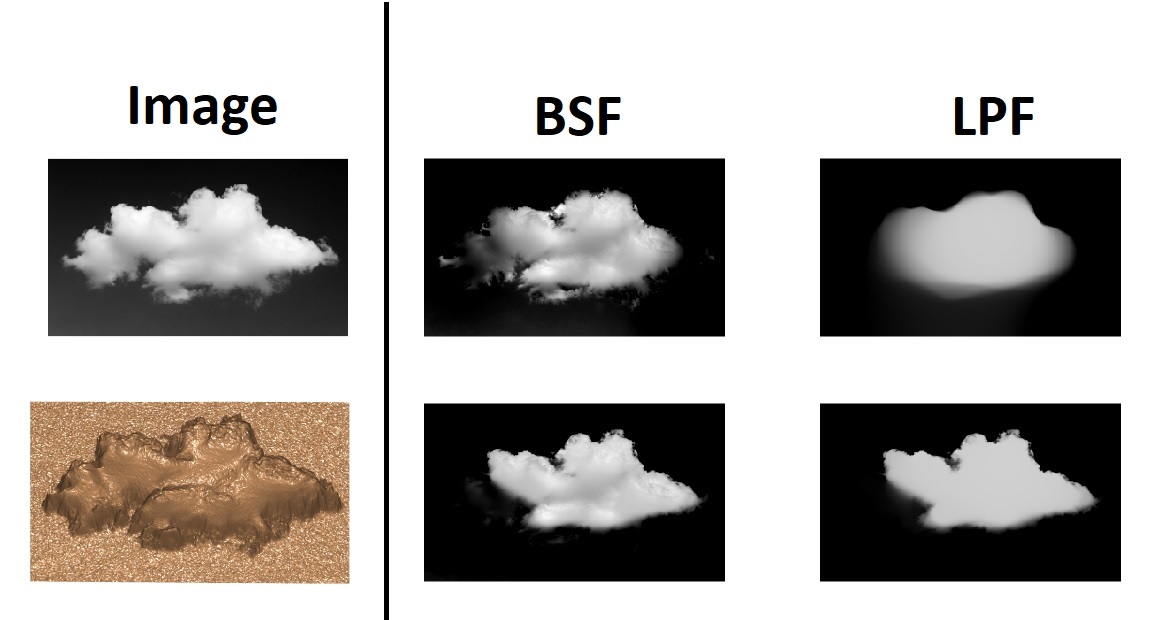}
\caption{Image Spectral $TV$ (upper row) and Spectral $ATV$ (bottom row) filtering, implemented as special cases of $M3$. Unlike the usual Spectral $ATV$, here the metric is induced by an explicit surface given by the image function (bottom left). Left-most column: Low-pass filter (LPF), demonstrating the difference between the "convexifying" effect of Euclidean spectral $TV$ compared to the minimal perimeter effect of the non-Euclidean spectral $NETV$ on surfaces. Middle column: Band stop filter (BSF), namely high-pass details where added to the LPF, resulting with a natural looking image}
\label{cloud_im}
\end{figure}

\begin{figure*} [htbp]
\includegraphics[width=\textwidth]{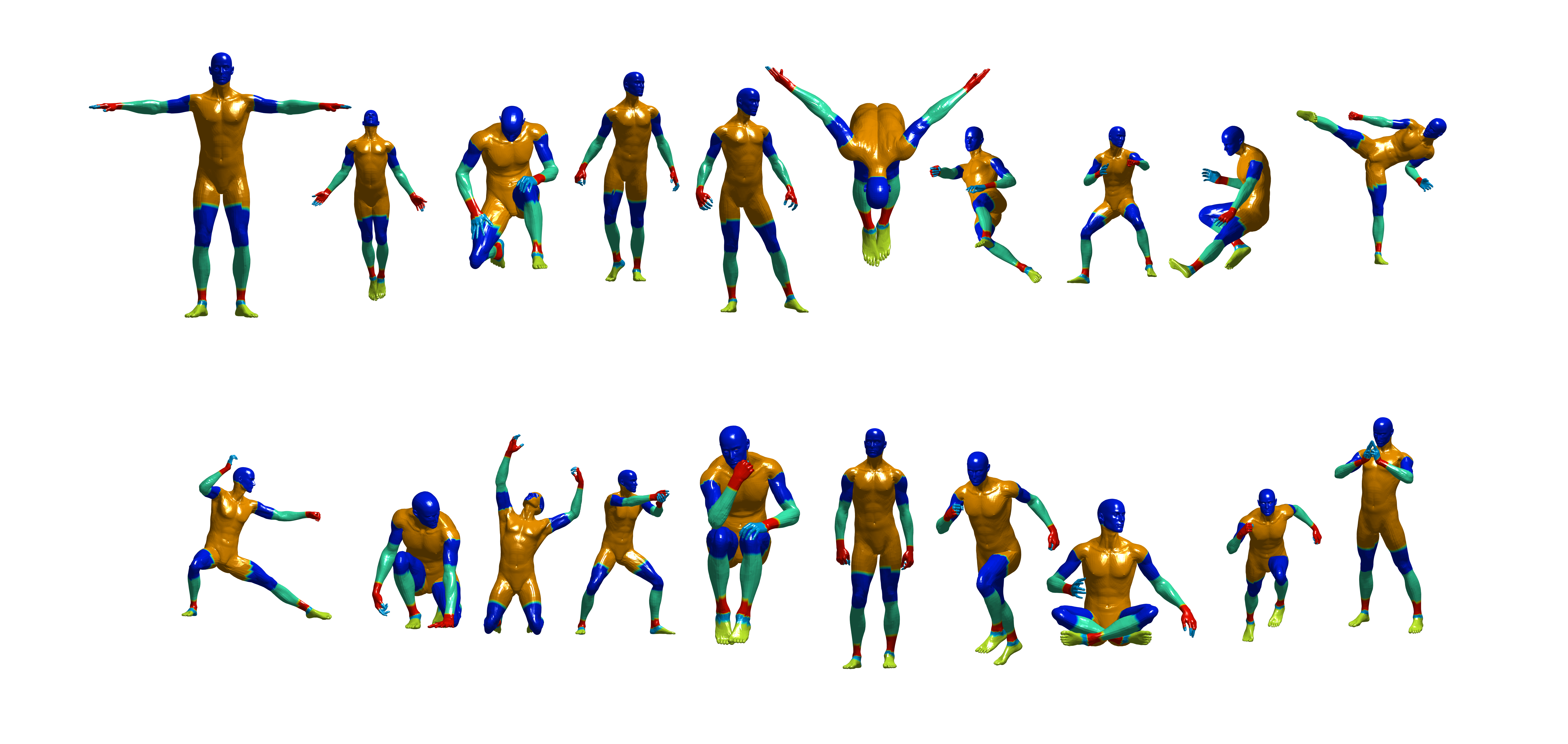}
\caption{Representations that enbale good feature separating filters are informative, hence may play a role in other tasks. Here we demonstrate an initial attempt at using $M2$'s spectral representation for co-segmentation of the Michael models in various  poses: Each point (vertex) on each pose (mesh) is assigned the representation $\phi(t)|_{q \in M}$. For co-segmentation we partitioned representation space using all representations from all poses - yielding segments that are shared across poses. Space partitioning was done via the k-means algorithm, using earth-movers distance instead of Euclidean distance. Earth-movers is a more appropriate measure since it is better at quantifying perturbations along the $t$ dimension.}
\label{michaels_seg}
\end{figure*}

\footnotetext{\href{https://www.thingiverse.com/thing:1700233}{Snail model}, courtesy of \href{https://3d-mon.com/}{3Demon} team. Licensed under the \href{https://creativecommons.org/licenses/by-sa/3.0/}{Creative Commons - Attribution - Share Alike}.}

\subsection{Flow-induced spectral representation considerations}
Our filtering framework is flow-based, where spectral representations are manifested as linear decaying components. Shape analysis is performed along the time domain. In contrast, the Laplace-Beltrami eigenfunctions are acquired by solving an eigenvalue problem on the non-Euclidean domain. One advantage of the flow-based framework is computational: The numerical simulation of a shape flow is often computationally cheaper than a numerical solution to an eigenvalue problem performed on the shape domain (e.g. solving the SVD of a discrete Laplace-Beltrami operator) - as demonstrated in Fig. \ref{armadil0 filtering}. This results with better filtering complexity. This is opposed to the Euclidean case - where eigenvalue decomposition of the Laplacian is not needed, since it is fixed and known (Fourier basis). Simulating a flow requires discrete time steps, resulting with a trade-off between computational complexity and simulation accuracy, controlled by the time step size. This is true for Euclidean and non-Euclidean settings alike. To measure the accuracy of the flow, we can test the reconstruction error of an all pass filter, i.e. $H(t)=1\forall t$, as demonstrated in Fig. \ref{fig: rec err}. Remark: Our semi-implicit implementation of the flows exhibits stability for large time-steps, even when they result with inaccuracies. This is attributed to the time-steps of both the filtering flows and the deformation flows being approximated as solutions to stable L2 minimization problems (similar stability may be found for instance in \cite{kazhdan2012can}).


Another consideration arises when comparing two representations, $\phi_1(t), \phi_2(t)$. The usual $L_2$ distance may not be compatible on time domains, and even more-so on discretized time domains. This is since $L_2$ does not express well the difference between small and large shifts in time. For example, consider an eigenfunction - which was shown to have a representation $\phi(t)=\delta(t-\frac{1}{\lambda})$, where $\lambda$ is the eigenvalue. Thus the $L_2$ measure will not be able to discriminate between eigenfunctions with similar eigenvalues to eigenfunctions with largely different eigenvalues. 
Thus appropriate distance measure must differentiate small from large time perturbations. One such measure is the earth movers' distance. An example using this distance on $\phi$ is portrayed in Fig. \ref{michaels_seg}, where co-segmentation takes place.

\subsection{Spectral image TV and shape spectral TV relation}
Our shape spectral $TV$ methods require a zero-homogeneous flow performed on a fixed metric. Considering the Beltrami and TV flow equivalence presented in \cite{kimmel1998image} (see a brief reminder of this equivalence in the appendix,  Sec. \ref{app:flow_equiv}), we have that applying spectral TV to images is a form of shape spectral TV. Let us have a closer look at this statement: To transition from MCF to the equivalent TV flow, the flow was rephrased on a fixed metric (the Euclidean pixel grid), which was absorbed as a nonlinearity of the operator, making it zero-homogeneous. Thus all requirements of the zero-homogeneous spectral framework were met.

With that said, this framework is more restrictive than our general framework in 3 ways: The operator suggested is one of a kind, the shape has to be parameterized as an image function, and the fixed metric is a Euclidean domain (the pixel-grid).

In $ATV$ the latter constraint is not required. As shown in \cite{biton2022adaptive}, $ATV$ can be re-interpreted using a generalization of the gradient , $\nabla_A f = A(x)\nabla f$. This gradient generalization can be obtained by considering the pixel grid as the $\Omega$ domain in Eq. \eqref{eq: patmeteric manifold}, and A as the metric $g$ from Eq. \eqref{eq: metric}, induced by some unspecified  $M$. While other aspects of $ATV$ do not coincide with the differential geometry framework we use here, it is certainly related to our work. Interestingly, in $ATV$ the importance of parameterization domain is greater than in our framework, as the signal lies in $\Omega$, and gradients are on  $\Omega$ as well, mapped from an unspecified non Euclidean domain. In contrast, our signal lies on an explicit manifold $M$.

$M3$ generalizes both spectral $TV$ and $ATV$  in the following sense: Considering the form $S=(x,y,f(x,y))$, choosing $\hat{S}=(x,y,0)$, we have that $\hat{d}=\hat{z}$, and $\alpha|S-\hat{S}|=(0, 0, f(x,y))$. Now consider two options - option 1: $\hat{S}$ induces $g$, resulting with a Euclidean flow of an image function. Option 2: $S$ induces $g$, resulting with a non Euclidean flow of an image function on an adapted metric. This is a form of $ATV$, but with a specified non-Euclidean surface domain. See example in Fig. \ref{cloud_im}.

\subsection{Future Ideas}
Our flow-based spectral framework can easily be adapted to a wide collection of operators that assume the required homogeneity and a fixed metric. Inevitably - neural-networks come to mind, where homogeneity can be taken care of using normalization layers.

Our new notion of non-Euclidean convexity, the locally minimal perimeter, might have an appropriate generalization to graphs - which are also non-Euclidean. This probably involves extending our theory from parametric surfaces to Riemannian manifolds of general dimensions.

The representations we use for filtering may be used for other tasks, such as classification and segmentation - see preliminary result in Fig. \ref{michaels_seg}.

Additonal key aspects from the Euclidean case may be generalized to our parametric surface setting, e.g. curvature bounds of eigensets, and analysis of vectorial functions.

\section{Summary}

We presented new nonlinear spectral theoretical analysis for surfaces, by generalizing nonlinear spectral theory of image processing. 
Based on our analysis, we proposed a general methodology for shape analysis and processing via nonlinear spectral filtering.

A key finding is our introduction of locally minimal perimeter sets, a novel generalization of conex sets to manifolds. It is derived by generalizing  properties of $NETV$ eigenfunctions. Our analysis is supported by numerical examples of minimizing flows, where numerical validation of eigenfunctions is performed by examining the decay near extinction, following the theory of \cite{bungert2020asymptotic}.

For shape nonlinear filtering our methods extract spectral representations from smoothing flows which satisfy two requirements: zero-homogeneity and a fixed metric. We choose to process the shape in its embedding space, providing unmediated nonlinear spectral representations, yielding good feature control. To showcase the general concept, three methods are proposed, where all three are based on the same mechanism, described in Eqs. \eqref{eq: phi}, \eqref{eq: nonlinear filtering}. Each method holds clear distinct properties induced by its flow, allowing various shape manipulations via spectral filtering. While possessing visibly distinct properties, all three methods demonstrate good smoothing and detail enhancement capabilities. Robustness to pose variations is demonstrated as well. With respect to processing time, we note that these methods are fairly fast, as they do not require solving an eigen-problem explicitly. Additionally, we present a Total-Variation approach for addressing the shape deformation problem. Our experiments show, that the deformation using our method is concentrated on plausible segment boundaries. Moreover - we have shown for several numerical cases, that these boundaries relate to our theoretical findings. \footnote{Models: Bust of Queen Nefertiti. Ägyptisches Museum und Papyrussammlung; Meteorite scan, courtesy of Jeremy Davis, Department of Art and Design, Central Michigan University; Trigon art; Stanford armadillo and poses by Belyaev, Yoshizawa, Seidel (2006);  Michaels from \cite{bronstein2008numerical}; various models from LIRIS database, Knubbel, Snail created by Thingiverse user 3D-mon, Cheburashka by Ilya Baran}



\bibliographystyle{ACM-Reference-Format}
\bibliography{sample-base}

\appendix

\section{Complementary Experiments}

\begin{figure} [htbp]
\includegraphics[width=0.5\textwidth]{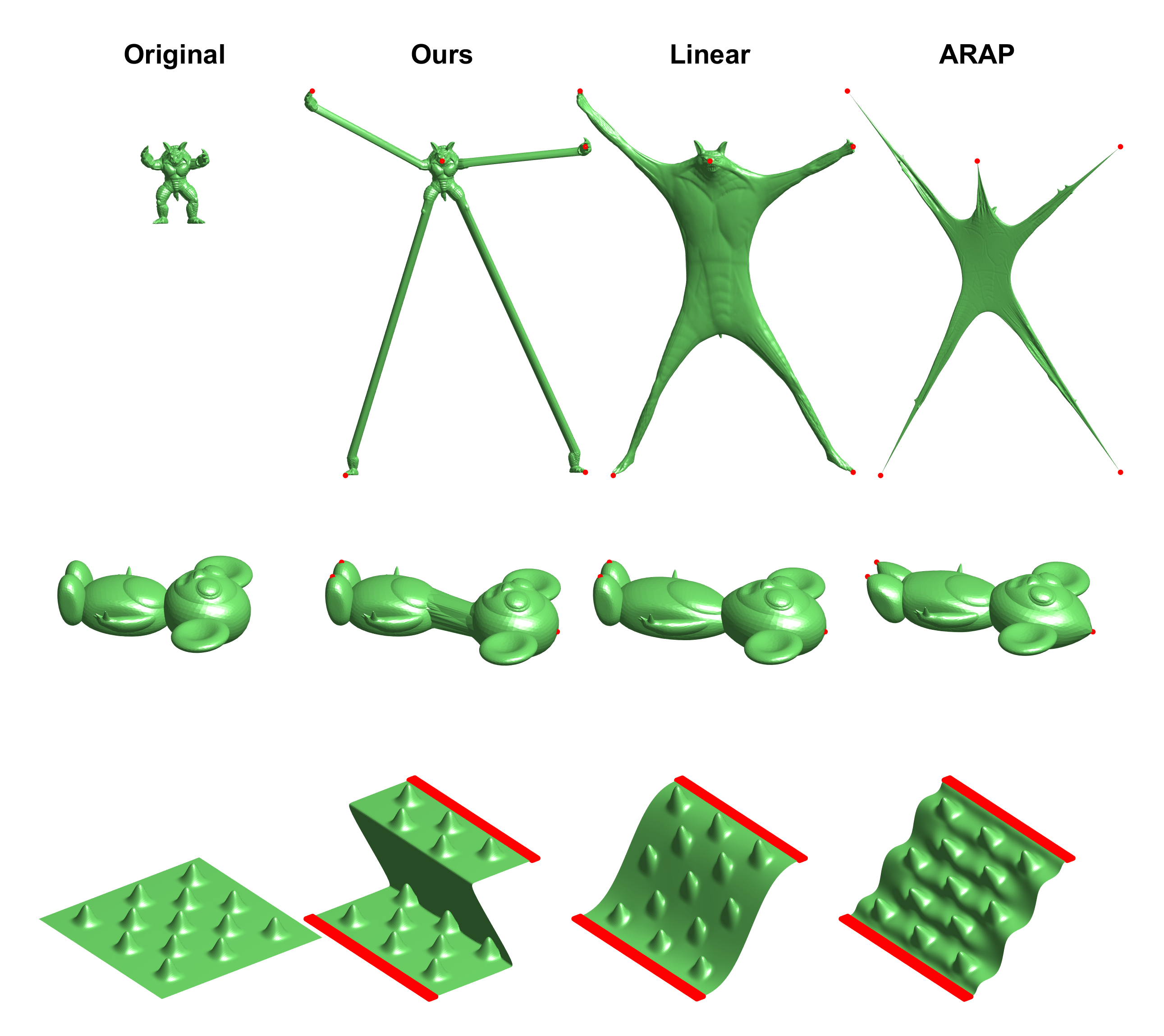}
\caption{The same experiment of Fig. \ref{fig:snail}, for the other shapes shown in Fig. \ref{fig:deformation_process}. As before the shape and deformation constraints are tested for the three methods. The same properties observed in Fig. \ref{fig:snail}can be seen here as well: The deformation field is piece-wise constant, and the deformation is concentrated on boundaries of plausible shape segments.}\label{fig:additional_stretch}
\end{figure}

Here we add experiments for completeness. Fig. \ref{sinc_z} bridges a gap beweeen Figs. \ref{sinc_set} and \ref{sinc}. Fig. \ref{fig:4_armadil_caricature} shows shape exaggeration similar to Fig. \ref{armadil0 filtering} on additional poses while not using smoothing. Fig. \ref{fig:additional_stretch} extends the method compariso performed for the Snail model in Fig. \ref{fig:snail} to three other models.
\begin{figure} [htbp]
\includegraphics[width=0.5\textwidth]{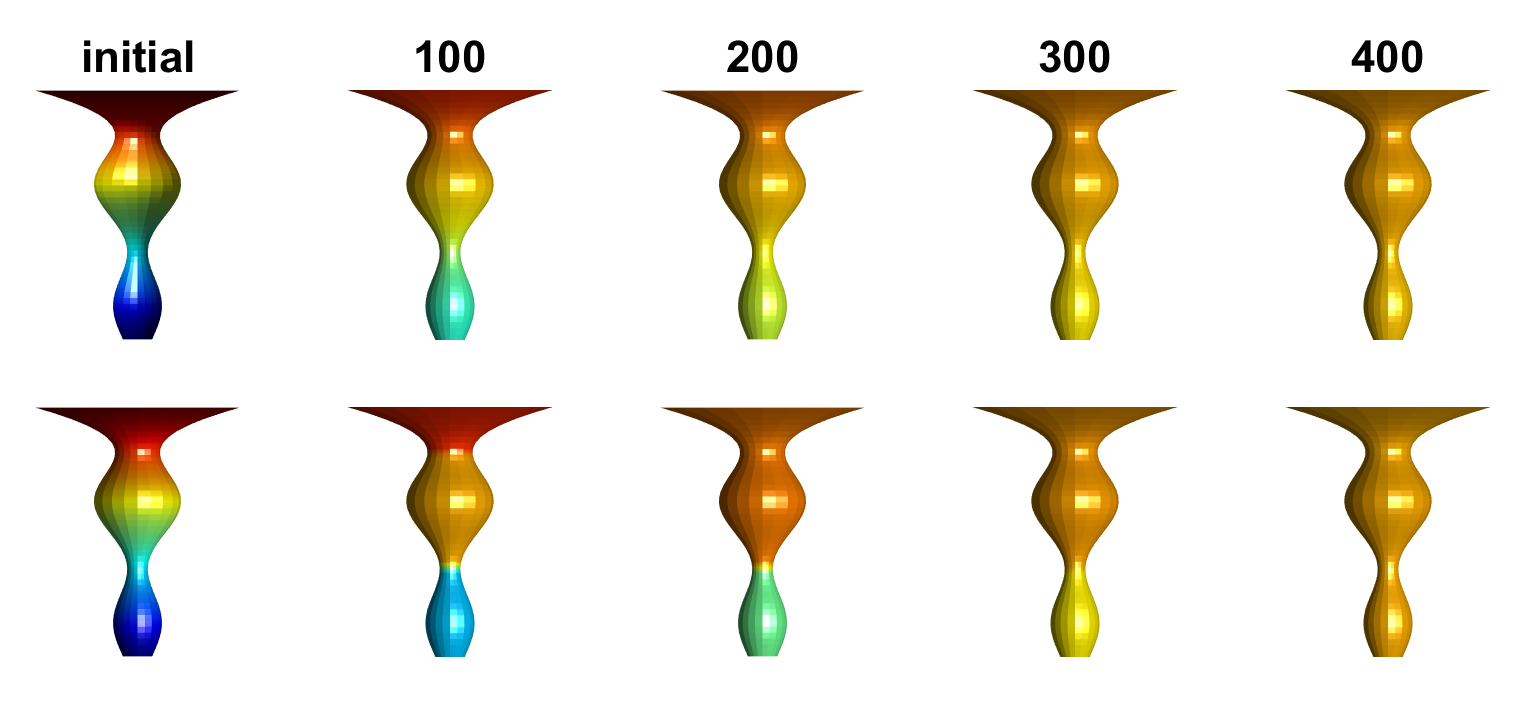}

\caption{A similar experiment to the one shown in Fig. \ref{sinc_set}, comparing linear diffusion (upper row) to the $NETV$ minimizing flow (bottom row). This experiment is performed using the  z-coordinate function. The resulting flow was used to induce the spectral representations filtered in  Fig. \ref{sinc}.}
\label{sinc_z}
\end{figure}

\begin{figure} [htbp]
\includegraphics[width=0.5\textwidth]{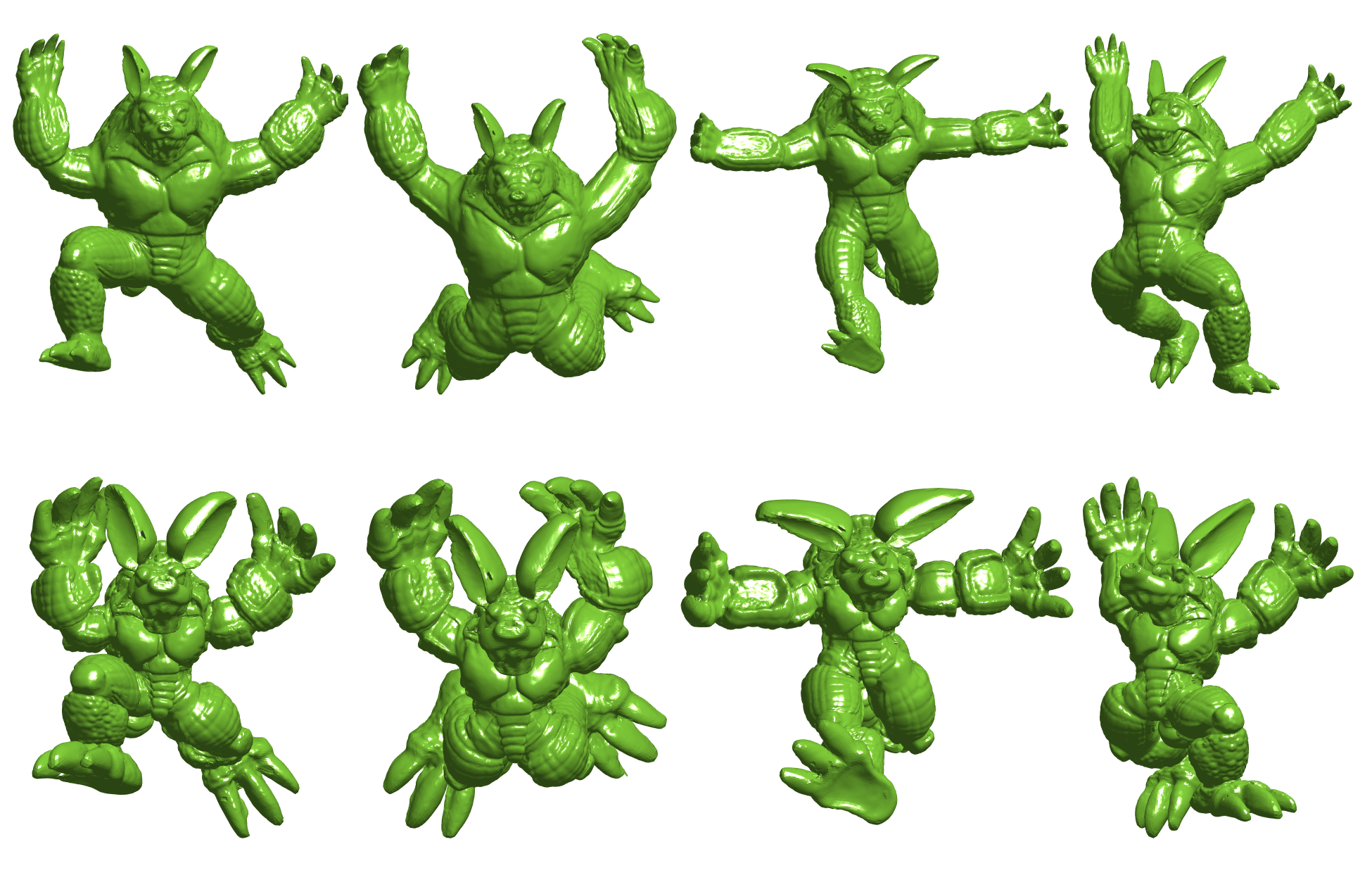}
\caption{Band-Pass exaggeration via $M3$ for the armadillo in various poses, where all poses experience the same filter. The band-pass exaggeration magnitude is identical to Fig. \ref{armadil0 filtering}. Unlike Fig. \ref{armadil0 filtering}, there is no smoothing of the high frequencies. The consistency across poses shows that our method can be considered as a charicaturization method by the properties posed in \cite{sela2015computational}.}
\label{fig:4_armadil_caricature}
\end{figure}


\section{Some proofs regarding Sec. \ref{sec: indicator theory}}\label{app:indicator_proofs}

\begin{claim*}
Let a field $z$ that is normal to the boundary of $C$ on (almost all) boundary points, and of norm less than or equal to one everywhere on $M$, i.e.
 \begin{equation}
 z=\tilde{n}^C \,\,for \,a.e.\, \omega_1, \omega_2 \in \partial \tilde{C},\,\,\,
 ||z||_g \leq 1 \,\forall \omega_1, \omega_2 \in \Omega,
 \end{equation}
where $a.e.$ stands for "almost every". Then $z$ admits the supremum of Eq. \eqref{eq: NETV of C} for $NETV(C)$\footnote{Requiring this almost everywhere on $\partial \tilde{C}$ is enough, as our proofs use $z$ under integration.}.
 \end{claim*}

\begin{proof} 
Since $\tilde{\chi}^C$ is an indicator of $\tilde{C}$ we have
\begin{equation} \label{eq: proof one}
\int_\Omega \nabla_g \cdot z \tilde{\chi}^C \, da = \int_{\tilde{C}} \nabla_g \cdot z\, da.
\end{equation}
Using manifold divergence Thm., as stated in Eq.\eqref{eq: par div th}, we have
\begin{equation} \label{mid res 0}
 \int_\Omega \nabla_g \cdot z \tilde{\chi}^C \, da =\int_{t_1}^{t_2} \langle z, \tilde{n}^C \rangle_g||\tilde{\gamma}^C_t||_g\, dt,
\end{equation}
where $\tilde{\gamma}^C_t=\frac{d\tilde{\gamma}^C}{dt}$. For convenience - let us reformulate this as
\begin{equation} \label{eq: mid res}
 \int_\Omega \nabla_g \cdot z \tilde{\chi}^C \, da =\int_{t_1}^{t_2}||z||_gcos\theta||\tilde{\gamma}^C_t||_g\, dt,
\end{equation}
where $\theta = \sphericalangle_g(z, \tilde{n}^C)=cos^{-1}\frac{\langle \tilde{n}^C, \tilde{z}\rangle_g}{||\tilde{n}^C||_g||\tilde{z}||_g}$. Maximization under the constraint  $||z||_g\leq1$ is achieved for  $\theta=0$, and $||z||_g=1$ almost everywhere on the boundary, i.e.

\begin{equation} \label{eq: normals solution}
 z=\tilde{n}^C\,\,for\,a.e.\, \omega_1, \omega_2 \in \partial \tilde{C}\Rightarrow NETV(C) = \int_\Omega \nabla_g \cdot z \tilde{\chi}^C \, da.
\end{equation}
\end{proof}
Remark: The "almost all" condition allows robustness to a zero-measure subset of $\partial C$ in which boundary normals are not defined (namely points of non-differentiable $\partial C$). In such a case  $\tilde{n}^C$ may be extended to satisfy
\begin{equation} \label{eq: n tilde}
    \begin{cases}
      J \tilde{n}^C = n^C \,\forall \omega_1, \omega_2 \in \partial \tilde{C} & \text{if}\ n^C \,\text{exists}\\
      ||\tilde{n}^C||_g \leq 1 & \text{otherwise},
    \end{cases}
  \end{equation}
  while keeping the proof intact.
\begin{claim*}  

\begin{equation} \label{eq: indictr NETV}
NETV(C)=per(C).
\end{equation}
\end{claim*}
This claim can be found for instance in  \cite{fumero2020nonlinear}. Let us re-prove it in our setting:

\begin{proof}
By Eq. \eqref{mid res 0}  we have

\begin{equation}
NETV(C) = \sup_{z}\int_{t_1}^{t_2} \langle z, \tilde{n}^C \rangle_g||\tilde{\gamma}^C_t||_g\, dt	\,\, s.t. \,\,||z||_g \leq 1 \,\forall \omega_1, \omega_2 \in \Omega.
\end{equation}
By Eq. \eqref{eq: normals solution} we can obtain a supremum by assigning $z=\tilde{n}^C\,\forall \omega_1, \omega_2 \in \partial \tilde{C}$ and have
\begin{equation}\label{eq: NETV as perimeter}
NETV(C)= \int_{t_1}^{t_2} ||\tilde{n}^C||^2_g\,||\tilde{\gamma}^C_t||_g\, dt = \int_{t_1}^{t_2} ||\tilde{\gamma}^C_t||_g\, dt=per(C),
\end{equation}
where the second equality uses $||\tilde{n}^C||_g=1 \,\forall \omega_1, \omega_2$, and  the last equality comes from Def. \eqref{eq: boundary length}.
 \end{proof}

\section{Sleeve sets as eigensets on the Torus}
Here we show that sleeve sets of the torus are eigensets, considering $M$ to be a torus with big and small radiis $R$, $r$.
\subsection{Torus preliminaries}
We choose its parametric formulation $S$ as follows: $S(\omega_1, \omega_2) =$ $\left((R+rcos(\omega_1))cos(\omega_2), (R+rcos(\omega_1))sin(\omega_2), r\,sin(\omega_1)\right)^T$, where $\omega_1,\omega_2 \in \Omega$ and $\Omega= [-\pi,\pi) \times [-\pi,\pi)$, inducing a metric $g(\omega_1, \omega_2) = \left(\begin{matrix}
\left(R + r\,cos(\omega_1)\right)^2 & 0 \\
0 & r^2
\end{matrix}\right)$, resulting with 
\begin{equation}
\sqrt{|g|} =\left(R + r\,cos(\omega_1)\right)r .    
\end{equation}

Let $z$ be a field on the torus, then its squared norm function  is
\begin{equation}
||z||_g =\sqrt{ \left(R + r\,cos(\omega_1)\right)^2\, z[1] + r^2\,z[2]},
\end{equation}

where $z[1], \,z[2]$ are components of the field. Remembering the divergence formula
\begin{equation} \label{eq:ddiivv}
\nabla_g\cdot \tilde{F} = \frac{1}{\sqrt{|g|}}\nabla_{\omega_1, \omega_2}\cdot(\sqrt{|g|}\tilde{F}),
\end{equation}
(see for instance \cite{do2016differential}) we have for the torus:

\begin{equation}
\nabla_g \cdot z = \frac{1}{\left(R + r\,cos(\omega_1)\right) \cancel{r}}\nabla_{\omega_1,\omega_2} \cdot \left[\left(R + r\,cos(\omega_1)\right) \cancel{r} z\right]
\end{equation}

\begin{equation} \label{eq: torus diver}
\nabla_g \cdot z = \frac{\partial }{\partial  \omega_1} z[1] + \frac{-r\,sin(\omega_1)}{R + r\,cos(\omega_1)} \, z[1] + \frac{\partial }{\partial  \omega_2}   z[2].
\end{equation}
The sleeve set of angle-length $l$, and center at $\omega_2 = c_0$ is denoted in parameterization domain as $\tilde{C}$, and defined as follows:
\begin{equation}
\tilde{C} = \{\omega_1, \omega_2: |\omega_2 - c_0 |\leq \frac{l}{2} \}.
\end{equation}
For convenience, W.T.L.O.G. we consider the sleeve set to have center at $\omega_2=0$, i.e.
\begin{equation}
\tilde{C} = \{\omega_1, \omega_2: |\omega_2 |\leq \frac{l}{2} \},
\end{equation}
resulting with
\begin{equation}\label{eq:torus sleeve boundary}
\partial \tilde{C} = \{\omega_1,\omega_2:\,|\omega_2|=\frac{l}{2}\}.
\end{equation}

\subsection{Finding a field of required properties}
The first property we are looking for is orthogonality to the boundary, on all boundary points. Since $g$ is diagonal, we have that $S$ preserves angles - hence orthogonality may be tested in parameterization domain. Furthermore, we need the field to be of unit norm on the boundary. By \eqref{eq:torus sleeve boundary} we have that $z$ is orthogonal to the boundary and of unit norm on the boundary if 
$z(\omega_1,\omega_2=\frac{l}{2})[1]=\frac{1}{R+rcos(\omega_1)}$,  $z(\omega_1,\omega_2=\frac{-l}{2})[1]=-\frac{1}{R+rcos(\omega_1)}$
, and
$z(\omega_1,|\omega_2=\frac{l}{2})[2]=0\,\forall \omega_1$.

Let $\Theta(\omega_2) =\begin{cases}
\omega_2-\frac{l}{2} & \omega_2 \in [0,\pi)]\\
\omega_2+\frac{l}{2} & \omega_2 \in (-\pi,0)
\end{cases} $. The choice
\begin{equation} \label{eq:torus good z}
z = \left(\frac{\frac{2}{l}\omega_2}{R+rcos(\omega_1)},\,\alpha 
\Theta(\omega_2)
\right)^T,
\end{equation}
satisfies above properties. Another required property is a constant divergence inside $C$. Let us show that this choice satisfies that as well, except for a zero-measure set of points: Plugging to \eqref{eq: torus diver} we have
\begin{equation}
\nabla_g \cdot z = \cancel{-\frac{2}{l}\omega_2\frac{-r\,sin(\omega_1)}{\left(R + r\,cos(\omega_1)\right)^2}} + \cancel{\frac{-r\,sin(\omega_1)}{R + r\,cos(\omega_1)} \, \frac{\frac{2}{l}\omega_2}{R+rcos(\omega_1)}} + \alpha\frac{\partial  
\Theta(\omega_2)}{\partial  \omega_2},   
\end{equation}
where $\frac{\partial  
\Theta(\omega_2)}{\partial  \omega_2} $ $=1\forall \omega_2 \neq 0,\pi$, thus
\begin{equation}
\nabla_g \cdot z = \alpha \,\forall \omega_2 \neq 0,\pi.
\end{equation}
\subsection{Demonstrating eigensets}
To prove an eigenset, we need to construct a field $\xi$, that satisfies the eigenfunction properties \eqref{def: eigenfunction} for a $\psi$ as in \eqref{eq: psi def}. In the current case $\psi$ is defined for a sleeve set $C$ on a torus.

First we note, that the appropriate field  $z$ of \eqref{eq:torus good z} can be similarly defined for the complement set - which is a sleeve set as well, but with its center at $\omega_2 = \pi$. It turns out, that the fields for center at either $\omega_2=\pi$, or at $\omega_2=0$, are of the same form, up to a sign factor. Thus we define $\xi$ to admit \eqref{eq:torus good z} for $\tilde{C}$ inside $\tilde{C}$, and the minus of \eqref{eq:torus good z} for $\Omega \backslash \tilde{C}$, i.e.
\begin{equation}
\xi = \begin{cases}
\left(\frac{\frac{2}{l}\omega_2}{R+rcos(\omega_1)},\,\alpha_1 
\Theta(\omega_2)
\right)^T  & \omega_1,\omega_2 \in \tilde{C}\\
\left(\frac{\frac{2}{l}\omega_2}{R+rcos(\omega_1)},\,\alpha_2 
\Theta(\omega_2)
\right)^T & \omega_1,\omega_2 \in \Omega \backslash \tilde{C}
\end{cases},
\end{equation}
which has $\nabla_g \cdot \xi = \begin{cases}
\alpha_1 & \omega_1,\omega_2 \in \tilde{C}\\
\alpha_2 & \omega_1,\omega_2 \in \Omega \backslash \tilde{C}
\end{cases}$, and is unit-orthogonal to the boundary, i.e. $z=\tilde{n}^C\,\forall \omega_1, \omega_2 \in \partial \tilde{C}$. If we set $\alpha_1=1,\,\alpha_2=\beta$, where $\beta$ is as in \eqref{eq: beta}, then indeed we have the eigenfunction property $\psi = \lambda \nabla_g \cdot \xi$.

By \eqref{def: eigenfunction} it is left to show that $\nabla_g \cdot \xi \in \partial NETV(\psi)$. To show this, it is sufficient, by Eq. \eqref{eq: converse subdiff}, to show that $\int_\Omega \nabla_g \cdot \xi \psi \, da =NETV(\psi)$. Let us begin: By one-homogeneity of the $NETV$ - we have that $NETV(\psi)=(1+\beta)NETV(C)$.

Since  $\xi=\tilde{n}^C \forall \omega_1, \omega_2 \in \partial \tilde{C}$, we have by Eq. \eqref{eq: normals solution} that $NETV(C) = \int_\Omega \nabla_g \cdot \xi \tilde{\chi}^C \, da$. 

Thus by the definition of $\psi$ \eqref{eq: psi def}, we have

$\int_\Omega \nabla_g \cdot \xi \psi \, da = $

$\int_\Omega \nabla_g \cdot \xi  (\tilde{\chi}^C - \beta \tilde{\chi}^{{M\backslash C}}) \, da = $

$\int_\Omega \nabla_g \cdot \xi ( \tilde{\chi}^C - \beta (1-\tilde{\chi}^{C})) \, da =$

$(1+\beta)\int_\Omega \nabla_g \cdot \xi \tilde{\chi}^C \, da - \beta \int_\Omega \nabla_g \cdot \xi  \, da = $

$(1+\beta)\int_\Omega \nabla_g \cdot \xi \tilde{\chi}^C \, da - \cancel{\lambda \beta \int_\Omega \psi  \, da} = $

$(1+\beta)\int_\Omega \nabla_g \cdot \xi \tilde{\chi}^C \, da$,

where the last cancellation uses the Neumann boundary condition assumption.  Thus we have $\int_\Omega \nabla_g \cdot \xi \psi \, da = (1+\beta)\int_\Omega \nabla_g \cdot \xi \tilde{\chi}^C \, da=(1+\beta)NETV(C)=NETV(\psi)$. By Eq. \eqref{eq: converse subdiff} we know that this is sufficient for $\xi \in \partial NETV(\psi)$.

\section{Shape Deformation Details}\label{app:shape_deformation}

First, let us consider for simplicity the scalar constrained $NETV$ flow: Consider a manifold $M$ and a function $f(q)$ on the manifold. Let us perform a total variation minimzing flow, initialized with $f$. The evolving function is constrained throughout the flow, i.e. denote $u(q, t)$ as the function at time $t$ of the flow, and the constrains $\{u(q_i, t)=c(q_i),\,q_i\in M\}_{i=1}^n$. Denote $u(q):=u(q,t=\infty)$.




Similarly, we may assume a surface $M$, with $x$ as its coordinate functions. The shape deformation problem requires to find some new surface with coordinates $x'$, which preserves the natural structure of $x$,  under the constraints that some points are pre-determined.  To this end we define the deformation field of a proposed $x'$ as $u=x'-x$ (from which $x'$ is reconstructed as $x'=x+u$). We perform constrained total-variation minimization on $u$ , where each nonzero constraint in  $\{c(q_i)\}_{i=1}^n$ is translated to a point $x'(q_i) \neq x(q_i)$. For the minimization process, we adapt a vectorial version of the process introduced by \cite{bronstein2016consistent}, where we use $p_{M1}$ of \eqref{eq:m1_operator} as the sub-differential. See Figs. \ref{fig:snail}, \ref{fig:deformation_process}, \ref{fig:additional_stretch} for the results.

\section{Shape flows and their equivalence sets}\label{app:flow_equiv}
This section serves as a reminder of observations on equivalent flows, namely the equivalence presented in \cite{kimmel1998image}. Consider the flow equation,
\begin{equation}
\frac{\partial f}{\partial t} =  p(f(t)),
\end{equation}
where $p$ is some operator. E.g., choosing $p=\Delta_{g,2}$ we obtain Eq. \eqref{eq: lin diff}. Suppose that $M$ is a manifold of some shape, and we would like to process this shape via a flow. To do so we can initialize the flow with the shape's coordinate function, i.e.  $f(0)=S$. Doing so with Eq. \eqref{eq: lin diff} is a classical shape smoothing flow. During the shape flow, at a given time $t$, we have  $p(f(t)): M \rightarrow \R^3$ , a field which describes the velocity of the evolving shape's points. Recalling the normal
\begin{equation} \label{Eq: normal}
N = \frac{S_{\omega_1} \times S_{\omega_2}}{|S_{\omega_1} \times  S_{\omega_2}|},
\end{equation}
a vector $\vec{V}\in p(f(t))$ may be decomposed to its normal and tangential component as  follows: $\vec{V}_N =  \langle\vec{V}, \vec{N}\rangle \vec{N}, \vec{V}_T= \vec{V}-\vec{V}_N$,  yielding$\frac{\partial f}{\partial t} =  p(f)_N+ p(f)_T$,  The normal component accounts for the change of the shape in time, and the tangential movement is merely a change of the shape's parameterization in time. Thus two shape flows are considered equal if their normal components are equal, and an equivalence set of shape flows is defined as:
\begin{equation} \label{Eq: equivalent flows}
\{\frac{\partial f}{\partial t} = q(f) : \langle q(f), \vec{N}\rangle = \langle p(f), \vec{N} \rangle \}.
\end{equation}
We note that parameterization constraints may cause differences between two flows from the same equivalence set.
\subsection{Beltrami -TV Flow equivalence}
In \cite{kimmel1998image}, an equivalence between MCF and the TV-flow was shown:
On one hand they define the Beltrami flow $\frac{\partial f}{\partial t} = \frac{H(t)}{<\vec{N}(t), \hat{z}>}\hat{z}$, which is obviously equivalent to MCF, in the sense of Eq. \eqref{Eq: equivalent flows}. On the other hand, consider a shape parameterized as an "image function", i.e. $S = (u, v, f(u,v))$, where the image is given by $f(u,v)$, and $u, v$ are a 2D Euclidean domain discretized as the pixel grid. In this case they show that plugging Eq. \eqref{Eq: normal} in to the Beltrami flow is equivalent to the image TV-flow.

Remark: equivalence by Eq. \eqref{Eq: equivalent flows} does not account for parameterization, which may induce implicit constraints - as is the case here: During Beltrami flow, the evolving shape's points are constrained to move in the $\hat{z}$ direction, thus keeping the parameterization $S = (u, v, f(u,v))$, contrary to MCF, where no such constraint exists. This is the reason MCF convergenes to a point,  while image TV flow converges to a plane $(u, v, const)$.

\end{document}